\documentclass[11pt,aps,pra,notitlepage,nofootinbib,tightenlines]{revtex4-1}

\usepackage[T1]{fontenc}
\usepackage[english]{babel}
\usepackage[utf8]{inputenc}

\usepackage[dvipsnames]{xcolor}
\usepackage{amsmath}
\usepackage{amssymb}
\usepackage{amsthm}
\usepackage{mathtools}
\usepackage[colorlinks=true,linkcolor=blue,citecolor=blue]{hyperref}
\usepackage{enumitem}
\setlist[enumerate]{noitemsep,partopsep=0pt,parsep=0pt}
\setlist[itemize]{noitemsep,partopsep=0pt,parsep=0pt}
\usepackage{booktabs}
\usepackage{nicematrix}

\newtheorem{thm}{Theorem}
\newtheorem{defn}{Definition}
\newtheorem{lem}[thm]{Lemma}
\newtheorem{cor}[thm]{Corollary}
\newtheorem{prop}[thm]{Proposition}
\theoremstyle{remark}
\newtheorem{rem}{Remark}

\numberwithin{equation}{section}

\DeclareMathOperator{\tr}{tr}

\DeclareMathOperator{\spec}{spec}

\DeclareMathOperator{\supp}{supp}
\DeclareMathOperator*{\argmin}{\arg\min}
\DeclareMathOperator*{\argmax}{\arg\max}
\newcommand{\ket}[1]{|#1\rangle}
\newcommand{\bra}[1]{\langle#1|}
\newcommand{\ketbra}[2]{\ket{#1}\!\bra{#2}}
\newcommand{\proj}[1]{\ketbra{#1}{#1}}
\newcommand{\braket}[2]{\langle#1|#2\rangle}
\newcommand{\brak}[1]{\braket{#1}{#1}}
\allowdisplaybreaks

\begin{document}

\title{Alternating minimization for computing doubly minimized Petz R\'enyi mutual information}
\author{Laura Burri}
\affiliation{Institute for Theoretical Physics, ETH Zurich, Zurich, Switzerland}

\begin{abstract}
The doubly minimized Petz R\'enyi mutual information (PRMI) of order $\alpha$ is defined as the minimization of the Petz divergence of order $\alpha$ of a fixed bipartite quantum state $\rho_{AB}$ relative to any product state $\sigma_A\otimes \tau_B$. To date, no closed-form expression for this measure has been found, necessitating the development of numerical methods for its computation. In this work, we show that alternating minimization over $\sigma_A$ and $\tau_B$ asymptotically converges to the doubly minimized PRMI for any $\alpha\in (\frac{1}{2},1)\cup (1,2]$, by proving linear convergence of the objective function values with respect to the number of iterations for $\alpha\in (1,2]$ and sublinear convergence for $\alpha\in (\frac{1}{2},1)$. Previous studies have only addressed the specific case where $\rho_{AB}$ is a classical-classical state, while our results hold for any quantum state $\rho_{AB}$.
\end{abstract}

\maketitle
%\tableofcontents

\section{Introduction}

The mutual information is a measure that quantifies the amount of correlation in a bipartite quantum state. 
Over the past decade, a number of R\'enyi generalizations of the mutual information have been introduced, 
such as generalizations based on the sandwiched divergence~\cite{mckinlay2020decomposition,
beigi2013sandwiched,leditzky2016strong,cheng2023tight,
wilde2014strong,leditzky2016strong,li2022reliability,li2024operational,burri2024properties2,
gupta2014multiplicativity,berta2015renyi,mosonyi2015coding,mosonyi2017strong,hayashi2016correlation}
and generalizations based on the Petz divergence~\cite{gupta2014multiplicativity,berta2015renyi,mosonyi2015coding,mosonyi2017strong,hayashi2016correlation,
kudlerflam2023renyi,kudlerflam2023renyi1,berta2021composite,burri2024properties}. 
This paper considers a member of the latter type of generalizations: the 
\emph{doubly minimized Petz R\'enyi mutual information} (PRMI)~\cite{berta2021composite,burri2024properties}.  
It is defined for a given bipartite quantum state $\rho_{AB}$ and $\alpha\in [0,\infty)$ as
\begin{align}\label{eq:def-2prmi}
I_\alpha^{\downarrow\downarrow}(A:B)_\rho 
\coloneqq \inf_{\sigma_A,\tau_B} D_\alpha (\rho_{AB}\| \sigma_A\otimes \tau_B),
\end{align}
where the minimization is over all quantum states $\sigma_A,\tau_B$, and $D_\alpha$ denotes the Petz divergence. 
This family of information measures is a one-parameter generalization of the mutual information, as the doubly minimized PRMI of order $\alpha=1$ coincides with the mutual information~\cite{hayashi2016correlation}.

In classical information theory, the analogous type of information measure was originally introduced in~\cite{tomamichel2018operational,lapidoth2019two}. 
The \emph{doubly minimized R\'enyi mutual information} (RMI) is defined for a given probability mass function (PMF) $P_{XY}$ and $\alpha\in [0,\infty)$ as
\begin{align}\label{eq:def-2rmi}
I_\alpha^{\downarrow\downarrow}(X:Y)_P
\coloneqq\inf_{Q_X,R_Y}D_\alpha (P_{XY}\| Q_X R_Y),
\end{align}
where the minimization is over all PMFs $Q_X,R_Y$, and $D_\alpha$ denotes the R\'enyi divergence. 
If $\rho_{AB}$ is a classical-classical (CC) state with PMF $P_{XY}$, then the doubly minimized PRMI reduces to the doubly minimized RMI, i.e., 
$I_\alpha^{\downarrow\downarrow}(A:B)_\rho =I_\alpha^{\downarrow\downarrow}(X:Y)_P$ for all $\alpha\in [0,\infty)$~\cite{burri2024properties}.
The doubly minimized RMI can therefore be regarded as a specific case of the doubly minimized PRMI.

To date, no closed-form expression for the doubly minimized PRMI of order $\alpha\neq 1$ that applies to any general $\rho_{AB}$ is known~\cite{burri2024properties}. 
A general closed-form expression is not available even for the case where $\rho_{AB}$ is a CC state, i.e., for the doubly minimized RMI of order $\alpha\neq 1$~\cite{lapidoth2019two}. 
The lack of methods for computing the doubly minimized PRMI hinders progress in understanding this measure (e.g., its relation to other information measures such as R\'enyi reflected entropies~\cite{burri2025minreflectedentropydoubly}). 
It is therefore useful to establish numerical methods for computing the doubly minimized PRMI. 
Since the doubly minimized PRMI of order $\alpha$ has an operational interpretation for $\alpha\in (\frac{1}{2},1)$ in the context of binary quantum state discrimination~\cite{burri2024properties}, numerical methods for this range of $\alpha$ are of particular importance.

Two approaches can be considered for computing the doubly minimized PRMI. 
The first approach uses the fact that 
$I_\alpha^{\downarrow\downarrow}(A:B)_\rho =\lim_{n\rightarrow\infty}\frac{1}{n}D_\alpha (\rho_{AB}^{\otimes n}\| \omega_{A^n}^n\otimes \omega_{B^n}^n)$ for any $\alpha\in [0,2]$, where $\omega_{A^n}^n$ and $\omega_{B^n}^n$ are certain universal permutation invariant states~\cite{burri2024properties}. 
Thus, the doubly minimized PRMI of order $\alpha\in [0,2]$ can be computed as the limit as $n\rightarrow\infty$ of 
$\frac{1}{n} D_\alpha (\rho_{AB}^{\otimes n}\| \omega_{A^n}^n\otimes \omega_{B^n}^n)$. 
However, as $n\rightarrow\infty$, the dimension of the Hilbert space $A^nB^n$ grows arbitrarily large, which suggests that computing 
$D_\alpha (\rho_{AB}^{\otimes n}\| \omega_{A^n}^n\otimes \omega_{B^n}^n)$ may become inefficient for large $n$. 
A straightforward implementation of the first approach for computing the doubly minimized PRMI is therefore not a viable option in practice. 
The second approach is to solve the optimization problem in~\eqref{eq:def-2prmi} numerically. 
This paper elaborates on this approach by employing the method of alternating minimization, where 
$D_\alpha(\rho_{AB}\| \sigma_A\otimes \tau_B)$ is minimized alternately over $\sigma_A$ and $\tau_B$.

\sloppy
For the classical setting, recent work~\cite{kamatsuka2024algorithms} has shown that alternating minimization of 
$D_\alpha (P_{XY}\| Q_XR_Y)$ over $Q_X$ and $R_Y$ asymptotically converges to the doubly minimized RMI for any $\alpha\in (\frac{1}{2},\infty)$. 
Subsequently, the non-asymptotic convergence has been analyzed: 
\cite{tsai2024linearconvergencehilbertsprojective} established linear convergence for \emph{the sequence of PMFs} obtained by alternating minimization to the global minimizer with respect to Hilbert's projective metric for any $\alpha\in (\frac{1}{2},\infty)$. 
Specifically, they show that the corresponding error rate is $\mathcal{O}(\lvert 1-\frac{1}{\alpha}\rvert^{2n})$, where $n$ denotes the number of iterations of alternating minimization. 
Their work does not provide explicit results on the convergence of the \emph{objective function values} to the doubly minimized RMI. 
It is, however, immediate to derive such results from their work, see Appendix~\ref{app:cc}. 
As we show there, their main result implies linear convergence of the objective function values to the doubly minimized RMI. 
Specifically, we show that the approximation error for the objective value is $\mathcal{O}(\lvert 1-\frac{1}{\alpha}\rvert^{2n})$.

This naturally raises the question of whether linear convergence be generalized from the classical to the quantum setting. 
In this work, we show that this is indeed the case for $\alpha\in (1,2]$ by utilizing similar techniques as~\cite{tsai2024linearconvergencehilbertsprojective}, applied to the cone of positive semidefinite operators instead of the cone of non-negative functions. 
We first show that a single iteration of alternating minimization is a contraction with respect to Hilbert's projective metric (Theorem~\ref{thm:linear}). 
As a corollary, this implies linear convergence for \emph{the sequence of quantum states} obtained by alternating minimization to the global minimizer with respect to Hilbert's projective metric for any $\alpha\in (1,2]$, and the corresponding error rate is $\mathcal{O}(\lvert 1-\frac{1}{\alpha}\rvert^{2n})$ (Corollary~\ref{cor:linear1}). 
In turn, this corollary implies linear convergence of the \emph{objective function values} to the doubly minimized PRMI of order $\alpha\in (1,2]$, and the corresponding error is $\mathcal{O}(\lvert 1-\frac{1}{\alpha}\rvert^{2n})$ (Corollary~\ref{cor:linear2}). 
Based on Corollary~\ref{cor:linear2}, we propose an alternating minimization algorithm for computing the doubly minimized PRMI of order $\alpha\in (1,2]$ for any given $\rho_{AB}$ up to an arbitrarily small error. 
These findings on the \emph{linear convergence} for $\alpha\in (1,2]$ form our first main result.

Our second main result concerns the convergence of alternating minimization for $\alpha\in (\frac{1}{2},1)$. 
For this range of $\alpha$, we have not been able to extend the proof method in~\cite{tsai2024linearconvergencehilbertsprojective} from the classical to the quantum setting. 
Instead, we adopt a completely different method to study the non-asymptotic convergence of alternating minimization for $\alpha\in (\frac{1}{2},1)$. 
Our method is motivated by the fact that the optimization problem in~\eqref{eq:def-2prmi} is jointly convex in $\sigma_A$ and $\tau_B$ for any fixed $\alpha\in (\frac{1}{2},1)$~\cite{burri2024properties}. 
The doubly minimized PRMI of order $\alpha$ can therefore be regarded as the minimization of a jointly convex function with two arguments, whose domains are convex subsets of Banach spaces. 
For such settings, general results on the convergence of alternating minimization exist. 
However, existing non-asymptotic analyses of alternating minimization~\cite{both2021rate} rely on smoothness assumptions in addition to the joint convexity of the objective function. 
Verifying these assumptions is not straightforward for the minimization problem in~\eqref{eq:def-2prmi}, which prevents a direct application of these results. 
(For instance,~\cite[Theorem~3]{both2021rate} guarantees sublinear convergence of alternating minimization under the assumption that the partial derivatives of the objective function are Lipschitz continuous~\cite[Assumption~(A2)]{both2021rate}. 
If this condition is satisfied, \cite[Theorem~3]{both2021rate} implies that the approximation error for the objective value after $n$ iterations of alternating minimization is $\mathcal{O}(1/n)$.) 
Instead, we analyze the convergence of alternating minimization for $\alpha\in (\frac{1}{2},1)$ from scratch by leveraging two additional properties of the optimization problem, alongside the joint convexity of the objective function. 
(More precisely, for technical reasons, our proof does not directly use the joint convexity of $D_\alpha (\rho_{AB}\| \sigma_A\otimes \tau_B)$, but instead makes use of the joint concavity of $Q_\alpha (\rho_{AB}\| \sigma_A\otimes \tau_B)$~\cite{burri2024properties}, the trace term appearing in the definition of $D_\alpha (\rho_{AB}\| \sigma_A\otimes \tau_B)$.) 
First, the partial optimization problems in~\eqref{eq:def-2prmi} can be solved uniquely and explicitly for any $\alpha\in (0,\infty)$~\cite{gupta2014multiplicativity,hayashi2016correlation}. 
Second, the global minimizer of~\eqref{eq:def-2prmi} is unique for any $\alpha\in (\frac{1}{2},1]$~\cite{burri2024properties}. 
These properties allow us to analyze the convergence of alternating minimization for~\eqref{eq:def-2prmi} for any $\alpha\in (\frac{1}{2},1)$, leading to our second main result: Theorem~\ref{thm:sublinear}. 
In this theorem, we establish \emph{sublinear convergence} of the objective function values to the doubly minimized PRMI of order $\alpha\in (\frac{1}{2},1)$. 
In particular, we show that the approximation error is $\mathcal{O}(1/n)$. 
Based on Theorem~\ref{thm:sublinear}, we propose an alternating minimization algorithm for computing the doubly minimized PRMI of order $\alpha\in (\frac{1}{2},1)$ for any given $\rho_{AB}$ up to an arbitrarily small error.

\paragraph*{Outline.} The remainder of this paper is structured as follows. 
Section~\ref{sec:preliminaries} outlines some mathematical preliminaries. 
In Section~\ref{sec:main}, we present our results on the convergence of alternating minimization for computing the doubly minimized PRMI of order $\alpha\in (\frac{1}{2},1)\cup (1,2]$. 
We first prove linear convergence for $\alpha\in (1,2]$ (\ref{ssec:linear}), 
and then sublinear convergence for $\alpha\in (\frac{1}{2},1)$ (\ref{ssec:sublinear}). 
Section~\ref{sec:conclusion} summarizes our findings and outlines possible directions for future work.

\section{Preliminaries}\label{sec:preliminaries}

\subsection{Notation}
We take ``$\log$'' to refer to the natural logarithm. 
The set of natural numbers strictly less than 
$n\in \mathbb{N}$ is denoted by $[n] \coloneqq \{0,1,\dots,n-1\}$.

All Hilbert spaces are assumed to be finite-dimensional and over $\mathbb{C}$. 
The dimension of a Hilbert space $A$ is denoted by $d_A$. 
The tensor product of two Hilbert spaces, $A$ and $B$, is denoted by $A\otimes B$ and may also be written as $AB$ for simplicity. 
$\mathcal{L}(A,B)$ denotes the set of linear maps from $A$ to $B$, and we set $\mathcal{L}(A)\coloneqq \mathcal{L}(A,A)$. 
To keep the notation short, identity operators are sometimes omitted. 
For instance, for $X_A\in\mathcal{L}(A)$, ``$X_A$'' may be understood as $X_A\otimes 1_B\in \mathcal{L}(A\otimes B)$. 
The kernel and spectrum of $X\in \mathcal{L}(A)$ are denoted by $\ker(X)$ and $\spec(X),$ respectively. 
The support of $X\in \mathcal{L}(A)$ is defined as the orthogonal complement of the kernel of $X$. 
For $X,Y\in \mathcal{L}(A)$, $X\ll Y$ is true iff $\ker(Y)\subseteq \ker(X)$.
For $X,Y\in \mathcal{L}(A)$, $X\perp Y$ is true iff $XY=YX=0$.
For $X\in \mathcal{L}(A)$, $X\geq 0$ is true iff $X$ is positive semidefinite, 
and $X>0$ is true iff $X$ is positive definite. 
For self-adjoint operators $X,Y\in \mathcal{L}(A)$, $X\geq Y$ is true iff $X-Y\geq 0$. 
Furthermore, $X\sim Y$ is true iff $X$ and $Y$ have the same support. 
The trace of $X\in \mathcal{L}(A)$ is denoted as $\tr[X]$, and the partial trace over $A$ is denoted as $\tr_A$. 

For a positive semidefinite $X\in \mathcal{L}(A)$, $X^p$ is defined for $p\in \mathbb{R}$ by taking the power on the support of $X$. 
The adjoint of $X\in \mathcal{L}(A,B)$ with respect to the inner products of $A,B$ is denoted by $X^\dagger\in \mathcal{L}(B,A)$, 
and the operator absolute value is denoted by $\lvert X\rvert \coloneqq (X^\dagger X)^{1/2}$. 
The Schatten $p$-norm of $X\in \mathcal{L}(A,B)$ is defined as $\|X \|_p\coloneqq \tr[\lvert X\rvert^p]^{1/p}$ for $p\in [1,\infty)$, and as 
$\lVert X\rVert_{\infty}\coloneqq \sqrt{\max (\spec(X^\dagger X))}$ for $p=\infty$.
The Schatten $p$-quasi-norm is defined as $\|X \|_p\coloneqq \tr[\lvert X\rvert^p]^{1/p}$ for $p\in (0,1)$.

The set of quantum states on $A$ is $\mathcal{S}(A)\coloneqq\{\rho\in \mathcal{L}(A):\rho\geq 0,\tr[\rho]=1\}$. 
In addition, we use the following symbols for constrained variants of this set with respect to a given self-adjoint $X\in \mathcal{L}(A)$.
\begin{align}
\mathcal{S}_{\ll X}(A)
&\coloneqq \{\rho\in \mathcal{S}(A): \rho\ll X\}
\\
\mathcal{S}_{X\ll}(A)
&\coloneqq \{\rho\in \mathcal{S}(A): X\ll \rho\}
\\
\mathcal{S}_{\sim X}(A)
&\coloneqq \{\rho\in \mathcal{S}(A): \rho\sim X\}
\\
\mathcal{S}_{>0}(A)
&\coloneqq \{\rho\in \mathcal{S}(A): \rho>0\}
\end{align}

\subsection{Petz divergence}
In this section, we define the Petz divergence and the doubly minimized PRMI, and highlight some of their properties relevant to this work.

The \emph{Petz (quantum R\'enyi) divergence (of order $\alpha$)} is defined for $\alpha\in [0,1)\cup (1,\infty),\rho\in \mathcal{S}(A),$ and any positive semidefinite $\sigma\in \mathcal{L}(A)$ as~\cite{petz1986quasi}
\begin{equation}
D_\alpha (\rho\| \sigma)\coloneqq\frac{1}{\alpha -1} \log \tr [\rho^\alpha \sigma^{1-\alpha}]
\end{equation}
if $(\alpha <1\land \rho\not\perp\sigma)\lor \rho\ll \sigma$ and 
$D_\alpha (\rho\| \sigma)\coloneqq \infty$ else. 
Moreover, $D_1$ is defined as the limit of $D_\alpha$ for $\alpha\rightarrow 1$. 
For $\alpha\in (0,\infty)$, we define
$Q_\alpha (\rho\| \sigma)\coloneqq
\tr [\rho^\alpha \sigma^{1-\alpha}]$
for all $\rho,\sigma\in \mathcal{S}(A)$.

In connection with the Petz divergence, the following R\'enyi Pinsker inequality holds~\cite[Theorem~2.1]{carlen2017remainder}. 
For any $\rho,\sigma\in \mathcal{S}(A)$ and $\alpha\in [\frac{1}{2},1)$
\begin{align}\label{eq:pinsker}
\frac{1-\alpha}{4\alpha} \|\rho^\alpha-\sigma^\alpha \|_{1/\alpha}^2
\leq 1-Q_\alpha(\rho\| \sigma).
\end{align}

The Petz divergence satisfies the following quantum Sibson identity~\cite{hayashi2016correlation}. 
For any $\alpha\in (0,\infty),\rho_{AB}\in \mathcal{S}(AB),\sigma_A\in \mathcal{S}(A),\tau_B\in \mathcal{S}(B)$ such that $\rho_A\not\perp\sigma_A$ holds
\begin{align}\label{eq:sibson}
D_\alpha(\rho_{AB}\| \sigma_A\otimes \tau_B)
=D_\alpha (\rho_{AB}\| \sigma_A\otimes \hat{\tau}_B)
+D_\alpha (\hat{\tau}_B\| \tau_B)
\quad\text{where}\quad 
\hat{\tau}_B\coloneqq \frac{(\tr_A[\rho_{AB}^\alpha \sigma_A^{1-\alpha}])^{\frac{1}{\alpha}}}{\tr[(\tr_A[\rho_{AB}^\alpha \sigma_A^{1-\alpha}])^{\frac{1}{\alpha}}]}.
\end{align}

The \emph{doubly minimized PRMI} of order $\alpha\in [0,\infty)$ is defined for $\rho_{AB}\in \mathcal{S}(AB)$ as
\begin{align}\label{eq:def_2prmi}
I_\alpha^{\downarrow\downarrow}(A:B)_\rho 
\coloneqq \inf_{(\sigma_A,\tau_B)\in\mathcal{S}(A)\times \mathcal{S}(B) } D_\alpha (\rho_{AB}\| \sigma_A\otimes \tau_B).
\end{align}
Properties of the doubly minimized PRMI have been studied in~\cite{burri2024properties}. 

Due to the quantum Sibson identity in~\eqref{eq:sibson}, partial minimizers of the optimization problem in~\eqref{eq:def_2prmi} can be characterized as follows. 
For any $\rho_{AB}\in \mathcal{S}(AB),\sigma_A\in \mathcal{S}(A),\alpha\in (0,1)\cup (1,\infty)$ holds that if 
$(\alpha\in (0,1)\land \rho_A\not\perp\sigma_A)\lor \rho_A\ll \sigma_A$, then~\cite{gupta2014multiplicativity,hayashi2016correlation,burri2024properties}
\begin{align}\label{eq:tau-optimal}
\hat{\tau}_B\coloneqq\frac{(\tr_A[\rho_{AB}^\alpha \sigma_A^{1-\alpha}])^{\frac{1}{\alpha}}}{\tr[(\tr_A[\rho_{AB}^\alpha \sigma_A^{1-\alpha}])^{\frac{1}{\alpha}}]}
\end{align}
is such that $\argmin_{\tau_B\in \mathcal{S}(B)} D_\alpha (\rho_{AB}\| \sigma_A\otimes \tau_B)
=\{\hat{\tau}_B\}$ and 
\begin{align}\label{eq:inf-tau}
\inf_{\tau_B\in \mathcal{S}(B)} D_\alpha (\rho_{AB}\| \sigma_A\otimes \tau_B)
=D_\alpha (\rho_{AB}\| \sigma_A\otimes \hat{\tau}_B)
&=\frac{1}{\alpha-1}\log \|\tr_A[\rho_{AB}^\alpha \sigma_A^{1-\alpha}] \|_{1/\alpha}.
\end{align}

For $\alpha\in (\frac{1}{2},1]$, the optimization problem in~\eqref{eq:def_2prmi} has a unique minimizer: 
For any given $\alpha\in (\frac{1}{2},1]$, there exists $(\hat{\sigma}_A,\hat{\tau}_B)\in \mathcal{S}_{\sim \rho_A}(A)\times \mathcal{S}_{\sim \rho_B}(B)$ such that~\cite{burri2024properties}
\begin{align}\label{eq:unique}
\argmin_{(\sigma_A,\tau_B)\in\mathcal{S}(A)\times \mathcal{S}(B) } D_\alpha (\rho_{AB}\| \sigma_A\otimes \tau_B)
=\{(\hat{\sigma}_A,\hat{\tau}_B)\}.
\end{align}
In particular, for $\alpha=1$, we have~\cite{hayashi2016correlation,burri2024properties}
\begin{align}\label{eq:unique1}
\argmin_{(\sigma_A,\tau_B)\in\mathcal{S}(A)\times \mathcal{S}(B) } D_1 (\rho_{AB}\| \sigma_A\otimes \tau_B)
=\{(\rho_A,\rho_B)\}.
\end{align}

\subsection{Hilbert's projective metric}\label{ssec:hilbert}
In this section, we define Hilbert's projective metric, which plays a key role in our analysis of the convergence of alternating minimization for $\alpha\in (1,2]$. 

To define Hilbert's projective metric, we need to recall some concepts from the theory of cones~\cite{lemmens2012nonlinear,lemmens2013birkhoffs}. 
Let $V$ be a real vector space. 
A subset $C\subseteq V$ is called a \emph{cone} if 
$C$ is convex, $\lambda C\subseteq C$ for all $\lambda\in [0,\infty)$, and 
$C\cap (-C)=\{0\}$. 
From now on, suppose $C\subseteq V$ is a cone. 
The corresponding induced partial order on $V$ is defined by $x\leq_C y$ iff $y-x\in C$, for all $x,y\in V$. 
For $x,y\in C$, $y$ \emph{dominates} $x$ ($x\ll_C y$) if there exist $\alpha,\beta\in \mathbb{R}$ such that $\alpha y\leq_C x\leq_C \beta y$. 
In that case, we define
\begin{align}
M(x/y)\coloneqq \inf\{\beta \in \mathbb{R}: x\leq_C \beta y\}.
\end{align} 
$x,y\in C$ are \emph{equivalent} ($x\sim_C y$) if $y$ dominates $x$ and $x$ dominates $y$. 
This defines an equivalence relation on $C$. 

\emph{Hilbert's projective metric} is defined for $x,y\in C$ by
\begin{align}
d_H(x,y)\coloneqq \log (M(x/y)M(y/x))
\end{align}
if $x\sim_C y$ and $y\neq 0$, and $d_H(x,y)\coloneqq 0$ if $x=y=0$, and $d_H(x,y)\coloneqq \infty$ else. 
For our analysis of the convergence of alternating minimization for $\alpha\in (1,2]$, we will invoke several properties of Hilbert's projective metric that are enumerated in Appendix~\ref{app:hilbert}.

\emph{Quantum setting.} 
In the quantum setting, the relevant real vector space $V$ is given by the set of all self-adjoint linear operators on a Hilbert space $A$, i.e., 
$V=\{X\in \mathcal{L}(A):X^\dagger =X\}$. 
The cone $C$ is given by the subset of all positive semidefinite operators, i.e., 
$C=\{X\in \mathcal{L}(A):X\geq 0\}$. 
The dominance relation $X\ll_C Y$ is given by $X\ll Y$ for all $X,Y\in C$. 
The equivalence relation $X\sim_C Y$ is given by $X\sim Y$ for all $X,Y\in C$.
Moreover, for all $X,Y\in C\setminus \{0\}$ such that $X\ll Y$~\cite{reeb2011hilbert}
\begin{align}\label{eq:m-xy-quantum}
M(X/Y)=\lVert Y^{-1/2}XY^{-1/2}\rVert_\infty.
\end{align}

\section{Main results}\label{sec:main}

\emph{Problem formulation.} We are interested in whether alternating minimization of $D_\alpha (\rho_{AB}\| \sigma_A\otimes \tau_B)$ over states $\sigma_A$ and $\tau_B$ converges to the doubly minimized PRMI $I_\alpha^{\downarrow\downarrow}(A:B)_\rho$ as the number of iterations tends to infinity, and if so, how fast this convergence occurs. 
The following definition formalizes the alternating minimization problem under consideration. 
The central question can then be rephrased as whether $\lim_{n\rightarrow\infty}x_n= I_\alpha^{\downarrow\downarrow}(A:B)_\rho$, and if so, how fast the convergence occurs with respect to $n$. 

\begin{defn}[Alternating minimization: Iteration rule]\label{def:am-quantum}
For any given $\alpha\in (0,\infty)$, $\rho_{AB}\in \mathcal{S}(AB)$, we define the functions
\begin{subequations}\label{eq:def-N}
\begin{align}
\mathcal{N}_{A\rightarrow B}:\mathcal{S}_{\rho_A\ll}(A)\rightarrow\mathcal{S}_{\sim \rho_B}(B), \quad
&\sigma_A\mapsto \frac{(\tr_A[\rho_{AB}^\alpha \sigma_A^{1-\alpha}])^{\frac{1}{\alpha}}}{\tr[(\tr_A[\rho_{AB}^\alpha \sigma_A^{1-\alpha}])^{\frac{1}{\alpha}}]},
\\
\mathcal{N}_{B\rightarrow A}:\mathcal{S}_{\rho_B\ll}(B)\rightarrow\mathcal{S}_{\sim \rho_A}(A), \quad
&\tau_B\mapsto \frac{(\tr_B[\rho_{AB}^\alpha \tau_B^{1-\alpha}])^{\frac{1}{\alpha}}}{\tr[(\tr_B[\rho_{AB}^\alpha \tau_B^{1-\alpha}])^{\frac{1}{\alpha}}]},
\end{align}
\end{subequations}
and for any given $\sigma_A^{(0)}\in \mathcal{S}_{\rho_A\ll}(A)$, we define for all $n\in\mathbb{N}$
\begin{align}
\tau_B^{(n)}
&\coloneqq \mathcal{N}_{A\rightarrow B}(\sigma_A^{(n)}),\\
\sigma_A^{(n+1)}
&\coloneqq \mathcal{N}_{B\rightarrow A}(\tau_B^{(n)}),\\
x_{n}
&\coloneqq D_\alpha (\rho_{AB}\| \sigma_A^{(n)}\otimes \tau_B^{(n)}).
\end{align}
\end{defn}

\begin{rem}[Correctness of iteration rule]\label{rem:correctness}
The iteration rule specified in Definition~\ref{def:am-quantum} yields a sequence of states that realizes alternating minimization, 
as~\eqref{eq:tau-optimal} and~\eqref{eq:unique1} imply that 
$\sigma_A^{(n+1)}
\in \argmin_{\sigma_A\in \mathcal{S}(A)} D_\alpha (\rho_{AB}\| \sigma_A\otimes \tau_B^{(n)})$ 
and 
$\tau_B^{(n)}
\in \argmin_{\tau_B\in \mathcal{S}(B)} D_\alpha (\rho_{AB}\| \sigma_A^{(n)}\otimes \tau_B)$ 
for any fixed $\alpha\in (0,\infty)$ and $\rho_{AB}\in \mathcal{S}(AB)$, and all $n\in \mathbb{N}$. 
\end{rem}

\begin{rem}[Alternating minimization for $\alpha=1$ is trivial]
For $\alpha=1$, alternating minimization trivially reaches the global minimizer after one iteration. 
More formally, let $\alpha\coloneqq 1$, 
let $\sigma_A^{(0)}\in \mathcal{S}_{\rho_A\ll}(A)$, 
and let $\tau_B^{(n)},\sigma_A^{(n+1)},x_n$ be given by Definition~\ref{def:am-quantum} for all $n\in \mathbb{N}$. 
Then, for all $n\in \mathbb{N}$, 
$\tau^{(n)}_B=\rho_B$ and $\sigma_A^{(n+1)}=\rho_A$ due to~\eqref{eq:def-N} and 
$x_{n+1}=I_\alpha^{\downarrow\downarrow}(A:B)_\rho$ due to~\eqref{eq:unique1}.
\end{rem}

\begin{rem}[Alternating minimization for $\alpha\in (0,\frac{1}{2})$ does not necessarily converge to $I_\alpha^{\downarrow\downarrow}(A:B)_\rho$]
For $\alpha\in (0,\frac{1}{2})$, alternating minimization with an arbitrary initialization $\sigma_A^{(0)}\in \mathcal{S}_{\rho_A\ll}(A)$ does not generally converge to a global minimizer. 
For instance, let $d_B\coloneqq d_A$, 
let $\{\ket{a_x}_A\}_{x\in [d_A]},\{\ket{b_y}_B\}_{y\in [d_B]}$ be orthonormal bases for $A,B$, 
and let 
$\rho_{AB}\coloneqq \frac{1}{d_A} \sum_{x\in [d_A]} \proj{a_x,b_x}_{AB}$. 
Let $\sigma_A^{(0)}\coloneqq 1_A/d_A$. 
Then, the states in Definition~\ref{def:am-quantum} are 
$\tau_B^{(n)}=1_B/d_B$ and $\sigma_A^{(n+1)}=1_A/d_A$ for all $n\in \mathbb{N}$. 
Thus, $(1_A/d_A,1_B/d_B)$ is a fixed point of $(\mathcal{N}_{B\rightarrow A}\circ \mathcal{N}_{A\rightarrow B},\mathcal{N}_{A\rightarrow B}\circ \mathcal{N}_{B\rightarrow A})$. 
By construction, this fixed point is a partial minimizer of $\min_{(\sigma_A,\tau_B)\in \mathcal{S}(A)\times \mathcal{S}(B)}D_\alpha (\rho_{AB}\| \sigma_A\otimes \tau_B)$, see Remark~\ref{rem:correctness}. 
However, it is not a global minimizer because~\cite{burri2024properties}
\begin{align}
\min_{(\sigma_A,\tau_B)\in \mathcal{S}(A)\times \mathcal{S}(B)}D_\alpha (\rho_{AB}\| \sigma_A\otimes \tau_B)
=\frac{\alpha}{1-\alpha}H_\infty(A)_\rho=\frac{\alpha}{1-\alpha}\log d_A
\end{align}
is strictly less than 
$D_\alpha (\rho_{AB}\|1_A/d_A\otimes 1_B/d_B)=\log d_A$. 
\end{rem}

\begin{rem}[Restriction of domain of $\mathcal{N}_{A\rightarrow B}$ to $\mathcal{S}_{\sim \rho_A}(A)$]\label{rem:restriction}
Let $\alpha\in (0,\infty)$, $\rho_{AB}\in \mathcal{S}(AB)$, 
$\sigma_A\in \mathcal{S}_{\rho_A\ll}(A)$, and let 
$\tilde{\sigma}_A\coloneqq \rho_A^0 \sigma_A\rho_A^0/\tr[\rho_A^0 \sigma_A]\in \mathcal{S}_{\sim \rho_A}(A)$. 
Then, $\mathcal{N}_{A\rightarrow B}(\sigma_A)=\mathcal{N}_{A\rightarrow B}(\tilde{\sigma}_A)$. 
This line of reasoning shows that 
$\mathcal{N}_{A\rightarrow B}(\mathcal{S}_{\rho_A\ll}(A))=\mathcal{N}_{A\rightarrow B}(\mathcal{S}_{\sim \rho_A}(A))$. 
Thus, it suffices to restrict attention to the case where the input of $\mathcal{N}_{A\rightarrow B}$ lies in $\mathcal{S}_{\sim \rho_A}(A)\subseteq \mathcal{S}_{\rho_A\ll }(A)$. 
For simplicity, several of the following results are therefore restricted to the case where the input of $\mathcal{N}_{A\rightarrow B}$ lies in $\mathcal{S}_{\sim \rho_A}(A)$, 
and similarly, the input of $\mathcal{N}_{B\rightarrow A}$ lies in $\mathcal{S}_{\sim \rho_B}(B)$.
\end{rem}

\subsection{Linear convergence for $\alpha\in (1,2]$}\label{ssec:linear}

The following theorem is the main result of this section. 
Its proof is given in Appendix~\ref{app:linear}.
The theorem provides sufficient conditions for $\mathcal{N}_{A\rightarrow B}$ and $\mathcal{N}_{B\rightarrow A}$ to be contractions with respect to Hilbert's projective metric. 

\begin{thm}[Contraction]\label{thm:linear}
Let $\alpha\in (1,2]$, 
$\rho_{AB}\in \mathcal{S}(AB)$, 
and let $\mathcal{N}_{A\rightarrow B},\mathcal{N}_{B\rightarrow A}$ be given by Definition~\ref{def:am-quantum}. 
Let $\gamma \coloneqq \lvert 1-\frac{1}{\alpha}\rvert$. 
\begin{enumerate}[label=(\alph*)]
\item For all $\sigma_A,\tilde{\sigma}_A\in \mathcal{S}_{\sim\rho_A}(A)$ 
and all $\tau_B,\tilde{\tau}_B\in \mathcal{S}_{\sim\rho_B}(B)$
\begin{align}
d_H(\mathcal{N}_{A\rightarrow B}(\sigma_A), \mathcal{N}_{A\rightarrow B}(\tilde{\sigma}_A))
&\leq \gamma d_H(\sigma_A,\tilde{\sigma}_A),
\label{eq:thm-linear-sigma1}\\
d_H(\mathcal{N}_{B\rightarrow A}(\tau_B), \mathcal{N}_{B\rightarrow A}(\tilde{\tau}_B))
&\leq \gamma d_H(\tau_B,\tilde{\tau}_B).
\label{eq:thm-linear-tau1}
\end{align}
\item If $\rho_{AB}>0:$ 
For all $\sigma_A,\tilde{\sigma}_A\in \mathcal{S}_{\sim\rho_A}(A)$ 
and all $\tau_B,\tilde{\tau}_B\in \mathcal{S}_{\sim\rho_B}(B)$
\begin{align}
d_H(\mathcal{N}_{A\rightarrow B}(\sigma_A), \mathcal{N}_{A\rightarrow B}(\tilde{\sigma}_A))
&\leq \gamma \kappa d_H(\sigma_A,\tilde{\sigma}_A),
\label{eq:thm-linear-sigma2}
\\
d_H(\mathcal{N}_{B\rightarrow A}(\tau_B), \mathcal{N}_{B\rightarrow A}(\tilde{\tau}_B))
&\leq \gamma\kappa d_H(\tau_B,\tilde{\tau}_B)
\label{eq:thm-linear-tau2}
\end{align}
where $\kappa \coloneqq \tanh(\delta/4)\in [0,1)$ and
\begin{align}
\delta \coloneqq \log \max_{\substack{\ket{\phi},\ket{\psi}\in A:\\ \brak{\phi}=1,\brak{\psi}=1 }} 
&\lVert \bra{\phi}\rho_{AB}^\alpha \ket{\phi}^{-\frac{1}{2}} \bra{\psi}\rho_{AB}^\alpha \ket{\psi} \bra{\phi}\rho_{AB}^\alpha \ket{\phi}^{-\frac{1}{2}} \rVert_\infty 
\\
&\cdot
\lVert \bra{\psi}\rho_{AB}^\alpha \ket{\psi}^{-\frac{1}{2}} \bra{\phi}\rho_{AB}^\alpha \ket{\phi} \bra{\psi}\rho_{AB}^\alpha \ket{\psi}^{-\frac{1}{2}} \rVert_\infty.
\end{align}
\end{enumerate}
\end{thm}

\begin{rem}[Subdivision of Theorem~\ref{thm:linear}]
Note that the assertion in~(a) holds in general and guarantees a contraction coefficient of $\gamma\in [0,1)$. 
In contrast,~(b) only holds under the assumption that $\rho_{AB}>0$, but guarantees an improved contraction coefficient of $\gamma\kappa\in [0,\gamma)\subseteq [0,1)$. 
This possibility of an improved coefficient will also appear analogously in subsequent results.
\end{rem}

The following result is a corollary of Theorem~\ref{thm:linear}. 
Its proof is given in Appendix~\ref{app:linear1}. 
The corollary implies linear convergence of the sequence of states obtained by alternating minimization with respect to Hilbert's projective metric. 
More specifically,~(a) implies that the corresponding error rate is $\mathcal{O}(\lvert 1-\frac{1}{\alpha}\rvert^{2n})$.
\begin{cor}[Linear convergence of states]\label{cor:linear1}
\sloppy
Let $\alpha\in (1,2]$, 
$\rho_{AB}\in \mathcal{S}(AB)$, 
and let 
$(\hat{\sigma}_A,\hat{\tau}_B)\in \argmin_{(\sigma_A,\tau_B)\in \mathcal{S}(A)\times \mathcal{S}(B)}D_\alpha (\rho_{AB}\| \sigma_A\otimes \tau_B)$. 
Let $\sigma_A^{(0)}\in \mathcal{S}_{\sim \rho_A}(A)$ 
and let $\tau^{(n)}_B,\sigma^{(n+1)}_A$ be given by Definition~\ref{def:am-quantum} for all $n\in \mathbb{N}$. 
Let $\gamma \coloneqq \lvert 1-\frac{1}{\alpha}\rvert$.
\begin{enumerate}[label=(\alph*)]
\item For all $n\in \mathbb{N}$
\begin{align}
d_H(\sigma_A^{(n)},\hat{\sigma}_A) &\leq \gamma^{2n} d_H(\sigma_A^{(0)},\hat{\sigma}_A),
\\
d_H(\tau_B^{(n)},\hat{\tau}_B) &\leq \gamma^{2n+1} d_H(\sigma_A^{(0)},\hat{\sigma}_A).
\end{align}
\item If $\rho_{AB}>0:$ 
For all $n\in \mathbb{N}$
\begin{align}
d_H(\sigma_A^{(n)},\hat{\sigma}_A) &\leq (\gamma\kappa)^{2n} d_H(\sigma_A^{(0)},\hat{\sigma}_A),
\\
d_H(\tau_B^{(n)},\hat{\tau}_B) &\leq (\gamma\kappa)^{2n+1} d_H(\sigma_A^{(0)},\hat{\sigma}_A)
\end{align}
where $\kappa$ is defined as in Theorem~\ref{thm:linear}~(b).
\end{enumerate}
\end{cor}

The following result follows from Corollary~\ref{cor:linear1}, as shown in Appendix~\ref{app:linear2}. 
It implies linear convergence of the objective function values obtained by alternating minimization. 
More specifically,~(a) implies that 
$\lvert x_{n}-I_\alpha^{\downarrow\downarrow}(A:B)_\rho \rvert$ is $\mathcal{O}(\lvert 1- \frac{1}{\alpha}\rvert^{2n})$. 
\begin{cor}[Linear convergence of objective function values]\label{cor:linear2}
\sloppy
Let $\alpha\in (1,2]$, 
$\rho_{AB}\in \mathcal{S}(AB)$, 
$(\hat{\sigma}_A,\hat{\tau}_B)\in \argmin_{(\sigma_A,\tau_B)\in \mathcal{S}(A)\times \mathcal{S}(B)}D_\alpha (\rho_{AB}\| \sigma_A\otimes \tau_B)$. 
Let $\sigma_A^{(0)}\in \mathcal{S}_{\sim \rho_A}(A)$ and let $(x_n)_{n\in \mathbb{N}}$ be given by Definition~\ref{def:am-quantum}. 
Let $\gamma \coloneqq \lvert 1-\frac{1}{\alpha}\rvert$. 
\begin{enumerate}[label=(\alph*)]
\item For all $n\in \mathbb{N}$
\begin{align}
\lvert x_{n}-I_\alpha^{\downarrow\downarrow}(A:B)_\rho \rvert 
&\leq \frac{1}{\lvert \alpha-1\rvert} \left(
\exp\left(\lvert \alpha-1\rvert (1+\gamma)\gamma^{2n} d_H(\sigma_A^{(0)},\hat{\sigma}_A) \right)
-1\right).
\end{align}
\item If $\rho_{AB}>0:$ 
For all $n\in \mathbb{N}$
\begin{align}
\lvert x_{n}-I_\alpha^{\downarrow\downarrow}(A:B)_\rho \rvert 
&\leq \frac{1}{\lvert \alpha-1\rvert} \left(
\exp\left(\lvert \alpha-1\rvert (1+\gamma\kappa)(\gamma\kappa)^{2n} d_H(\sigma_A^{(0)},\hat{\sigma}_A) \right)
-1\right)
\end{align}
where $\kappa$ is defined as in Theorem~\ref{thm:linear}~(b).
\end{enumerate}
\end{cor}

\begin{rem}[Asymptotic convergence, uniqueness of minimizer]\label{rem:linear-asymptotic}
We can conclude from Corollary~\ref{cor:linear1}~(a) and Corollary~\ref{cor:linear2}~(a) that alternating minimization converges to the doubly minimized PRMI of order $\alpha$ for any $\alpha\in (1,2]$, and that the minimizer of the underlying optimization problem is unique. 
More formally, let $\alpha\in (1,2]$, 
$\rho_{AB}\in \mathcal{S}(AB)$, 
$\sigma_A^{(0)}\in \mathcal{S}_{\rho_A\ll}(A)$, 
and let $\tau_B^{(n)},\sigma_A^{(n+1)},x_n$ be given by Definition~\ref{def:am-quantum} for all $n\in \mathbb{N}$. 
Then, $(x_n)_{n\in \mathbb{N}}$ is monotonically decreasing and 
$\lim_{n\rightarrow \infty}x_n=I_\alpha^{\downarrow\downarrow}(A:B)_\rho$, 
the limits $\sigma_A^{(\infty)}\coloneqq \lim_{n\rightarrow\infty}\sigma_A^{(n)}$ and 
$\tau_B^{(\infty)}\coloneqq \lim_{n\rightarrow\infty}\tau_B^{(n)}$ exist, and
\begin{equation}
\argmin_{(\sigma_A,\tau_B)\in \mathcal{S}(A)\times\mathcal{S}(B)}D_\alpha (\rho_{AB}\| \sigma_A\otimes \tau_B)=\{(\sigma_A^{(\infty)},\tau_B^{(\infty)})\} 
\subseteq \mathcal{S}_{\sim \rho_A}(A)\times \mathcal{S}_{\sim \rho_B}(B).
\end{equation}
\end{rem}

\begin{rem}[Upper bound on $d_H(\sigma_A^{(0)},\hat{\sigma}_A)$ independent of $\hat{\sigma}_A$]\label{rem:dh-sigma}
The upper bounds in Corollary~\ref{cor:linear1} and Corollary~\ref{cor:linear2} depend on $\hat{\sigma}_A$ through the constant $d_H(\sigma_A^{(0)},\hat{\sigma}_A)$. 
From a computational standpoint, it is desirable to obtain upper bounds that do not depend on the global minimizer $\hat{\sigma}_A$, since $\hat{\sigma}_A$ is not known a priori. 
This issue is addressed in Appendix~\ref{app:dh-sigma}, where we derive an upper bound on $d_H(\sigma_A^{(0)},\hat{\sigma}_A)$ that is independent of $\hat{\sigma}_A$. 
\end{rem}

\begin{rem}[Alternating minimization algorithm for $\alpha\in {(1,2]}$]\label{rem:am}
Based on Corollary~\ref{cor:linear2}, we propose an alternating minimization algorithm for computing the doubly minimized PRMI of order $\alpha\in (1,2]$: Algorithm~1. 
For any given inputs 
$\alpha\in (1,2]$, $\rho_{AB}\in \mathcal{S}(AB)$, $\varepsilon_0\in (0,\infty)$, $\sigma_A^{(0)}\in \mathcal{S}_{\rho_A\ll}(A)$, 
Corollary~\ref{cor:linear2} and Remark~\ref{rem:dh-sigma} (see Proposition~\ref{prop:bound} in Appendix~\ref{app:dh-sigma}) imply that 
Algorithm~1 terminates after a finite number of steps, with the output $x$ satisfying 
$\lvert x-I_\alpha^{\downarrow\downarrow}(A:B)_\rho\rvert \leq \varepsilon_0$. 
The input $\sigma_A^{(0)}\in \mathcal{S}_{\rho_A\ll}(A)$ can, in principle, be chosen arbitrarily. 
However, since the optimal choice for $\alpha=1$ is $\rho_A$, see~\eqref{eq:unique1}, a particularly natural choice for the initial state is $\sigma_A^{(0)}\coloneqq \rho_A$.
\end{rem}

\begin{center}
\begin{NiceTabular}{ll}[code-before = \rowcolor{gray!20}{1}]
\toprule
\textbf{Algorithm~1}
&Output $x$ satisfies $\lvert x-I_\alpha^{\downarrow\downarrow}(A:B)_\rho \rvert \leq \varepsilon_0$ for $\alpha\in (1,2]$.
\\ \midrule
\textbf{Input} &
$\alpha\in (1,2],\rho_{AB}\in \mathcal{S}(AB),\varepsilon_0 \in (0,\infty),\sigma_A^{(0)}\in \mathcal{S}_{\rho_A\ll}(A)$
\\ \hline
\textbf{Definitions}&
$\mathcal{N}_{A\rightarrow B}$, $\mathcal{N}_{B\rightarrow A}$ as in~\eqref{eq:def-N}
\\
&$c_0$ as in~\eqref{eq:def-linear-c0}
\\
&$\gamma\coloneqq 1-\frac{1}{\alpha}$
\\ \hline
\textbf{Alternating} &
$\tau_B\coloneqq \mathcal{N}_{A\rightarrow B}(\sigma_A^{(0)})$; 
$x\coloneqq D_\alpha (\rho_{AB}\| \sigma_A^{(0)}\otimes \tau_B)$
\\
\textbf{minimization} &
$n\coloneqq 0$; 
$\varepsilon\coloneqq \frac{1}{\alpha-1} \left(
\exp((\alpha -1) (1+\gamma)\gamma^{2n} c_0 )
-1\right)$
\\
&while $\varepsilon \geq \varepsilon_0$:
\\
&\hspace*{2em}
$\sigma_A\coloneqq \mathcal{N}_{B\rightarrow A}(\tau_B)$; 
$\tau_B\coloneqq \mathcal{N}_{A\rightarrow B}(\sigma_A)$; 
$x\coloneqq D_\alpha (\rho_{AB}\| \sigma_A\otimes \tau_B)$
\\
&\hspace*{2em}
$n\coloneqq n+1$; 
$\varepsilon\coloneqq \frac{1}{\alpha -1} \left(
\exp((\alpha -1) (1+\gamma)\gamma^{2n} c_0 )
-1\right)$
\\ \hline
\textbf{Output}&
$x$
\\ \bottomrule
\end{NiceTabular}
\end{center}

\subsection{Sublinear convergence for $\alpha\in (\frac{1}{2},1)$}\label{ssec:sublinear}

The following theorem is the main result of this section. 
Its proof is given in Appendix~\ref{app:proof-frechet}. 

\begin{thm}[Sublinear convergence of objective function values]\label{thm:sublinear}
Let $\alpha\in (\frac{1}{2},1)$, 
$\rho_{AB}\in \mathcal{S}(AB)$, 
$\sigma_A^{(0)}\in \mathcal{S}_{\rho_A\ll}(A)$, and let $(x_n)_{n\in \mathbb{N}}$ be given by Definition~\ref{def:am-quantum}. 
Moreover, let
\begin{subequations} \label{eq:def-sublinear-c0}
\begin{align}
\lambda_A&\coloneqq \min(\spec(\tr_B[\rho_{AB}^\alpha])\setminus\{0\}),
\label{eq:def-lam_a}\\
\lambda_B&\coloneqq \min(\spec(\tr_A[\rho_{AB}^\alpha])\setminus\{0\}),
\\
\lambda_{A,0}&\coloneqq \min (\spec(\tilde{\sigma}_A^{(0)} )\setminus \{0\}) 
\quad \text{where}\quad 
\tilde{\sigma}_A^{(0)}\coloneqq \rho_A^0\sigma_A^{(0)}\rho_A^0/\tr[\rho_A^0\sigma_A^{(0)}],
\\
c_0
&\coloneqq 2\sqrt{5} \max(\lambda_B^{-1}, \lambda_A^{\frac{\alpha(1-\alpha)}{1-2\alpha}}\lambda_B^{\frac{\alpha^2}{1-2\alpha}} ) \lambda_{A,0}^{\alpha-1}.
\end{align}
\end{subequations} 
Then, for all $n\in \mathbb{N}_{>0}$ 
\begin{align}
\lvert x_{n}-I_\alpha^{\downarrow\downarrow}(A:B)_\rho\rvert
&\leq c_0 \sqrt{x_{n-1}-x_{n}},
\label{eq:thm-n1}
\end{align}
and as a consequence,
\begin{align}
\lvert x_{n}-I_\alpha^{\downarrow\downarrow}(A:B)_\rho\rvert
&\leq \max\left(\frac{3}{2} c_0^{2},2x_0\right) \frac{1}{n}.
\label{eq:thm-n2}
\end{align}
\end{thm}

\begin{rem}[Asymptotic convergence, uniqueness of minimizer]\label{rem:sublinear-asymptotic}
We can conclude from Theorem~\ref{thm:sublinear} that alternating minimization converges to the doubly minimized PRMI of order $\alpha$ for any $\alpha\in (\frac{1}{2},1)$. 
The underlying minimizer has been shown to be unique and to be contained in $\mathcal{S}_{\sim\rho_A}(A)\times \mathcal{S}_{\sim \rho_B}(B)$ in previous work~\cite{burri2024properties}, see~\eqref{eq:unique}. 
More formally, 
let $\alpha\in (\frac{1}{2},1)$, 
$\rho_{AB}\in \mathcal{S}(AB)$, 
$\sigma_A^{(0)}\in \mathcal{S}_{\rho_A\ll}(A)$, 
and let $\tau_B^{(n)},\sigma_A^{(n+1)},x_n$ be given by Definition~\ref{def:am-quantum} for all $n\in \mathbb{N}$. 
Then, $(x_n)_{n\in \mathbb{N}}$ is monotonically decreasing and 
$\lim_{n\rightarrow \infty}x_n=I_\alpha^{\downarrow\downarrow}(A:B)_\rho$, 
the limits $\sigma_A^{(\infty)}\coloneqq \lim_{n\rightarrow\infty}\sigma_A^{(n)}$ and 
$\tau_B^{(\infty)}\coloneqq \lim_{n\rightarrow\infty}\tau_B^{(n)}$ exist, and
\begin{equation}
\argmin_{(\sigma_A,\tau_B)\in \mathcal{S}(A)\times\mathcal{S}(B)}D_\alpha (\rho_{AB}\| \sigma_A\otimes \tau_B)=\{(\sigma_A^{(\infty)},\tau_B^{(\infty)})\} 
\subseteq \mathcal{S}_{\sim \rho_A}(A)\times \mathcal{S}_{\sim \rho_B}(B).
\end{equation}
\end{rem}

\begin{rem}[Alternating minimization algorithm for $\alpha\in (\frac{1}{2},1)$]
Based on Theorem~\ref{thm:sublinear}, we propose an alternating minimization algorithm for computing the doubly minimized PRMI of order $\alpha\in (\frac{1}{2},1)$: Algorithm~2. 
For any given inputs 
$\alpha\in (\frac{1}{2},1)$, $\rho_{AB}\in \mathcal{S}(AB)$, $\varepsilon_0\in (0,\infty)$, $\sigma_A^{(0)}\in \mathcal{S}_{\rho_A\ll}(A)$, 
Theorem~\ref{thm:sublinear} implies that 
Algorithm~2 terminates after a finite number of steps, with the output $x$ satisfying 
$\lvert x-I_\alpha^{\downarrow\downarrow}(A:B)_\rho\rvert \leq \varepsilon_0$.
The input $\sigma_A^{(0)}\in \mathcal{S}_{\rho_A\ll}(A)$ can be chosen arbitrarily, but a particularly natural choice is $\sigma_A^{(0)}\coloneqq \rho_A$ (cf. Remark~\ref{rem:am}).
\end{rem}

\begin{center}
\begin{NiceTabular}{ll}[code-before = \rowcolor{gray!20}{1}]
\toprule
\textbf{Algorithm~2}
&Output $x$ satisfies $\lvert x-I_\alpha^{\downarrow\downarrow}(A:B)_\rho \rvert \leq \varepsilon_0$ for $\alpha\in (\frac{1}{2},1)$.
\\ \midrule
\textbf{Input} &
$\alpha\in (\frac{1}{2},1),\rho_{AB}\in \mathcal{S}(AB),\varepsilon_0 \in (0,\infty),\sigma_A^{(0)}\in \mathcal{S}_{\rho_A\ll}(A)$
\\ \hline
\textbf{Definitions}&
$\mathcal{N}_{A\rightarrow B}$, $\mathcal{N}_{B\rightarrow A}$ as in~\eqref{eq:def-N}
\\
&$c_0$ as in~\eqref{eq:def-sublinear-c0}
\\ \hline
\textbf{Alternating} &
$\tau_B\coloneqq \mathcal{N}_{A\rightarrow B}(\sigma_A^{(0)})$; 
$x'\coloneqq D_\alpha (\rho_{AB}\| \sigma_A^{(0)}\otimes \tau_B)$
\\
\textbf{minimization} &
$\sigma_A\coloneqq \mathcal{N}_{B\rightarrow A}(\tau_B)$; 
$\tau_B\coloneqq \mathcal{N}_{A\rightarrow B}(\sigma_A)$; 
$x\coloneqq D_\alpha (\rho_{AB}\| \sigma_A\otimes \tau_B)$
\\ 
&$\varepsilon\coloneqq c_0 \sqrt{ x'-x }$
\\
&while $\varepsilon \geq \varepsilon_0$:
\\
&\hspace*{2em}
$x'\coloneqq x$
\\
&\hspace*{2em}
$\sigma_A\coloneqq \mathcal{N}_{B\rightarrow A}(\tau_B)$; 
$\tau_B\coloneqq \mathcal{N}_{A\rightarrow B}(\sigma_A)$; 
$x\coloneqq D_\alpha (\rho_{AB}\| \sigma_A\otimes \tau_B)$
\\
&\hspace*{2em}
$\varepsilon\coloneqq c_0 \sqrt{ x'-x }$
\\ \hline
\textbf{Output}&
$x$
\\ \bottomrule
\end{NiceTabular}
\end{center}

\section{Conclusion}\label{sec:conclusion}
We have analyzed the convergence of alternating minimization of $D_\alpha (\rho_{AB}\| \sigma_A\otimes \tau_B)$ over states $\sigma_A$ and $\tau_B$ to the doubly minimized PRMI $I_\alpha^{\downarrow\downarrow}(A:B)_\rho$ for any fixed $\alpha\in (\frac{1}{2},1)\cup (1,2]$ and $\rho_{AB}\in \mathcal{S}(AB)$. 
Our main results address the non-asymptotic convergence of the objective function values after $n$ iterations of alternating minimization, denoted by $x_n$. 
Specifically, we have shown that 
\begin{itemize}[noitemsep]
\item for $\alpha \in (1,2]$, 
$\lvert x_n-I_\alpha^{\downarrow\downarrow}(A:B)_\rho \rvert$ is 
$\mathcal{O}(\lvert 1-\frac{1}{\alpha}\rvert^{2n})$, see Corollary~\ref{cor:linear2}~(a), and
\item for $\alpha\in (\frac{1}{2},1)$, 
$\lvert x_n-I_\alpha^{\downarrow\downarrow}(A:B)_\rho \rvert$ is
$\mathcal{O}(1/n)$ as $n\rightarrow\infty$, see Theorem~\ref{thm:sublinear}. 
\end{itemize}
These results prove linear and sublinear convergence of the objective function value, respectively. 
Based on these results, we have proposed two alternating minimization algorithms for computing the doubly minimized PRMI of order $\alpha$ up to an arbitrarily small error, for $\alpha \in (1,2]$ and $\alpha \in (\frac{1}{2},1)$, respectively, see Algorithm~1 and Algorithm~2. 

As explained in the introduction, previous work~\cite{tsai2024linearconvergencehilbertsprojective} in classical information theory implies that if $\rho_{AB}$ is a CC state, then 
$\lvert x_n-I_\alpha^{\downarrow\downarrow}(A:B)_\rho \rvert$ is 
$\mathcal{O}(\lvert 1-\frac{1}{\alpha}\rvert^{2n})$ for any $\alpha\in (\frac{1}{2},1)\cup (1,\infty)$. 
Our findings show that this result can be generalized to all quantum states if $\alpha\in (1,2]$ by using similar proof techniques. 
Whether linear convergence holds for all quantum states over an even larger range of $\alpha$ remains an open question for future research.

This work focuses on alternating minimization as a specific numerical method for computing the doubly minimized PRMI. 
How alternating minimization compares to other numerical methods~\cite{bauschke2017convex,zaslavski2020projected,zaslavski2022optimization} in this context is left as an open question.

Another potential direction for future research is to analyze the convergence of alternating minimization for computing the doubly minimized \emph{sandwiched} R\'enyi mutual information (SRMI) of order $\alpha$~\cite{burri2024properties2}. 
This measure is defined in a similar way to the doubly minimized PRMI, except that the sandwiched divergence $\widetilde{D}_\alpha$ is used instead of the Petz divergence $D_\alpha$. 
It is known that $\widetilde{D}_\alpha(\rho_{AB}\| \sigma_A\otimes\tau_B)$ is jointly convex in $\sigma_A$ and $\tau_B$~\cite{cheng2023tight} for $\alpha\in (1,\infty)$ 
(just as $D_\alpha(\rho_{AB}\| \sigma_A\otimes\tau_B)$ is jointly convex for $\alpha\in (\frac{1}{2},1)$~\cite{burri2024properties}), suggesting that alternating minimization is a suitable numerical method for its computation~\cite{both2021rate}. 
However, analyzing the convergence of alternating minimization for computing the doubly minimized SRMI is likely to be more challenging than for the doubly minimized PRMI, since no explicit expression is currently known for the partial minimizers associated with the doubly minimized SRMI~\cite{hayashi2016correlation,burri2024properties2}. 

\begin{acknowledgments}
I am grateful to Renato Renner for helpful discussions and comments, and to Martin Sandfuchs for discussions. 
This work was supported by 
the Swiss National Science Foundation via project No.\ 20QU-1\_225171, 
the CHIST-ERA project No.\ 20CH21\_218782,
and
the National Centre of Competence in Research SwissMAP, 
and the Quantum Center at ETH Zurich.
\end{acknowledgments}

\appendix
\section{Classical setting}\label{app:cc}
In this appendix, we review the results on alternating minimization for computing the doubly minimized RMI in~\cite{tsai2024linearconvergencehilbertsprojective}. 
We then show how these results imply linear convergence of the objective function values to the doubly minimized RMI (Corollary~\ref{cor:cc-linear2}). 
To this end, we first explain our notation for the classical setting (\ref{ssec:notation-cc}). 
Then, we review the findings of~\cite{tsai2024linearconvergencehilbertsprojective} and extend them by the aforementioned corollary (\ref{ssec:linear-cc}).

\subsection{Notation for classical setting}\label{ssec:notation-cc}
\subsubsection{Probability mass functions}

Let $\mathcal{X}$ be a finite set.
The support of a function $P:\mathcal{X}\rightarrow\mathbb{R}$ is denoted as $\supp (P)\coloneqq\{x\in \mathcal{X}:P(x)\neq 0\}$.
For two functions $P,Q:\mathcal{X}\rightarrow\mathbb{R}$, 
$P\ll Q$ is true iff $\supp(P)\subseteq \supp(Q)$. 
$P\sim Q$ is true iff $\supp(P)=\supp(Q)$. 
$P\perp Q$ is true iff $\supp(P)\cap \supp(Q)=\emptyset$.

The set of PMFs over $\mathcal{X}$ is 
$\mathcal{P}(\mathcal{X})\coloneqq \{P:\mathcal{X}\rightarrow [0,\infty) \text{ such that } \sum_{x\in \mathcal{X}}P(x)=1\}$. 
The set of PMFs that have the same support as $P:\mathcal{X}\rightarrow\mathbb{R}$ is denoted as 
$\mathcal{P}_{\sim P}(\mathcal{X})\coloneqq \{Q\in \mathcal{P}(\mathcal{X}):Q\sim P\}$.

\subsubsection{R\'enyi divergence}
The \emph{R\'enyi divergence} of order $\alpha\in [0,1)\cup (1,\infty)$ is defined for 
$P,Q\in \mathcal{P}(\mathcal{X})$ as
\begin{align}
D_\alpha (P\| Q) \coloneqq \frac{1}{\alpha-1}\log \sum_{x\in \supp(P)} 
P(x)^\alpha Q(x)^{1-\alpha}
\end{align}
where for $\alpha > 1$, we read $P(x)^\alpha Q(x)^{1-\alpha}$ as $P(x)^\alpha /Q(x)^{\alpha-1}$ and we adopt the convention $y/0 = \infty$ for $y > 0$. 
Moreover, $D_1$ is defined as the limit of $D_\alpha$ for $\alpha\rightarrow 1$. 

The \emph{doubly minimized RMI} of order $\alpha\in [0,\infty)$ is defined for $P_{XY}\in \mathcal{P}(\mathcal{X}\times \mathcal{Y})$ as
\begin{align}
I_\alpha^{\downarrow\downarrow}(X:Y)_P \coloneqq 
\inf_{(Q_X,R_Y)\in \mathcal{P}(\mathcal{X})\times \mathcal{P}(\mathcal{Y})} D_\alpha (P_{XY}\| Q_X R_Y).
\label{eq:cc-2prmi}
\end{align}
Properties of the doubly minimized RMI have been studied in~\cite{lapidoth2019two}. 

Partial minimizers of the optimization problem in~\eqref{eq:cc-2prmi} can be characterized as follows. 
For any $P_{XY}\in \mathcal{P}(\mathcal{X}\times \mathcal{Y}),Q_X\in \mathcal{P}(\mathcal{X}),\alpha\in (0,\infty)$ holds that if 
$(\alpha\in (0,1)\land P_X\not\perp R_X)\lor P_X\ll R_X$, then~\cite{lapidoth2019two}
\begin{align}\label{eq:qx-optimal}
\hat{R}_Y(\cdot)
&\coloneqq \frac{(\sum_{x\in \mathcal{X}}P_{XY}(x,\cdot)^\alpha Q_X(x)^{1-\alpha} )^{\frac{1}{\alpha}}}{\sum_{y\in \mathcal{Y}}(\sum_{x\in \mathcal{X}}P_{XY}(x,y)^\alpha Q_X(x)^{1-\alpha} )^{\frac{1}{\alpha}}}
\end{align}
is such that $\argmin_{R_Y\in \mathcal{P}(\mathcal{Y})} D_\alpha (P_{XY}\| Q_X R_Y)
=\{\hat{R}_Y\}$.

For $\alpha\in (\frac{1}{2},\infty)$, the optimization problem in~\eqref{eq:cc-2prmi} has a unique minimizer: 
For any given $\alpha\in (\frac{1}{2},\infty)$, there exists $(\hat{Q}_X,\hat{R}_Y)\in \mathcal{P}_{\sim P_X}(X)\times \mathcal{P}_{\sim P_Y}(Y)$ such that~\cite{lapidoth2019two,burri2024properties}
\begin{align}\label{eq:cc-unique}
\argmin_{(Q_X,R_Y)\in\mathcal{P}(X)\times \mathcal{P}(Y)} D_\alpha (P_{XY}\| Q_X R_Y)
=\{(\hat{Q}_X,\hat{R}_Y)\}.
\end{align}

\subsubsection{Hilbert's projective metric}
Recall the notation associated with Hilbert's projective metric from Section~\ref{ssec:hilbert}.

\emph{Classical setting.} 
In the classical setting, the relevant real vector space $V$ is given by the set of all functions from a finite set $\mathcal{X}$ to $\mathbb{R}$. 
The cone $C$ is given by the subset of all functions from $\mathcal{X}$ to $[0,\infty)$. 
The dominance relation $P\ll_C Q$ is given by $P\ll Q$ for all $P,Q\in C$. 
The equivalence relation $P\sim_C Q$ is given by $P\sim Q$ for all $P,Q\in C$. 
Moreover, for all $P,Q\in C\setminus \{0\}$ such that $P\ll Q$
\begin{align}\label{eq:m-xy-classical}
M(P/Q)=\max_{x\in \supp(Q)}\frac{P(x)}{Q(x)}.
\end{align}

\subsection{Linear convergence for $\alpha\in (\frac{1}{2},\infty)$ for classical setting}\label{ssec:linear-cc}
Alternating minimization of $D_\alpha (P_{XY}\| Q_XR_Y)$ over PMFs $Q_X$ and $R_Y$ proceeds as described in the following definition, see~\eqref{eq:qx-optimal}.

\begin{defn}[Alternating minimization: Iteration rule]\label{def:am-classical}
For any given $\alpha\in (0,\infty)$, $P_{XY}\in \mathcal{P}(\mathcal{X}\times \mathcal{Y})$, we define the functions
\begin{subequations}\label{eq:def-N-cc}
\begin{align}
\mathcal{N}_{X\rightarrow Y}: \mathcal{P}_{P_X\ll }(\mathcal{X}) \rightarrow \mathcal{P}_{\sim P_Y}(\mathcal{Y}),
\quad
&Q_X\mapsto \frac{(\sum_{x\in \mathcal{X}}P_{XY}(x,\cdot)^\alpha Q_X(x)^{1-\alpha} )^{\frac{1}{\alpha}}}{\sum_{y\in \mathcal{Y}}(\sum_{x\in \mathcal{X}}P_{XY}(x,y)^\alpha Q_X(x)^{1-\alpha} )^{\frac{1}{\alpha}}},
\\
\mathcal{N}_{Y\rightarrow X}: \mathcal{P}_{P_Y\ll }(\mathcal{Y}) \rightarrow \mathcal{P}_{\sim P_X}(\mathcal{X}),
\quad
&R_Y\mapsto \frac{(\sum_{y\in \mathcal{Y}}P_{XY}(\cdot,y)^\alpha R_Y(y)^{1-\alpha} )^{\frac{1}{\alpha}}}{\sum_{x\in \mathcal{X}}(\sum_{y\in \mathcal{Y}}P_{XY}(x,y)^\alpha R_Y(y)^{1-\alpha} )^{\frac{1}{\alpha}}},
\end{align}
\end{subequations}
and for any given $Q_X^{(0)}\in \mathcal{P}_{P_X\ll }(\mathcal{X})$, we define for all $n\in \mathbb{N}$
\begin{align}
R_Y^{(n)} &\coloneqq \mathcal{N}_{X\rightarrow Y}(Q_X^{(n)}),
\\
Q_X^{(n+1)} &\coloneqq \mathcal{N}_{Y\rightarrow X}(R_Y^{(n)}),
\\
x_n &\coloneqq D_\alpha (P_{XY}\| Q_X^{(n)}R_Y^{(n)}).
\end{align}
\end{defn}

The following theorem provides sufficient conditions for $\mathcal{N}_{X\rightarrow Y}$ and $\mathcal{N}_{Y\rightarrow X}$ to be contractions with respect to Hilbert’s projective metric. 
\begin{thm}[Contraction]\label{thm:cc-linear}
\cite[Theorem~13]{tsai2024linearconvergencehilbertsprojective}
Let $\alpha\in [\frac{1}{2},1)\cup (1,\infty)$, 
$P_{XY}\in \mathcal{P}(\mathcal{X}\times \mathcal{Y})$, and 
let $\mathcal{N}_{X\rightarrow Y},\mathcal{N}_{Y\rightarrow X}$ be given by Definition~\ref{def:am-classical}. 
Let $\gamma \coloneqq \lvert 1-\frac{1}{\alpha}\rvert$.
\begin{enumerate}[label=(\alph*)]
\item 
If $\alpha >\frac{1}{2}$: 
For all $Q_X,\tilde{Q}_X\in \mathcal{P}(\mathcal{X}),R_Y,\tilde{R}_Y\in \mathcal{P}(\mathcal{Y})$
\begin{align}
d_H(\mathcal{N}_{X\rightarrow Y}(Q_X),\mathcal{N}_{X\rightarrow Y}(\tilde{Q}_X))
&\leq \gamma d_H(Q_X,\tilde{Q}_X) ,
\\
d_H(\mathcal{N}_{Y\rightarrow X}(R_Y),\mathcal{N}_{Y\rightarrow X}(\tilde{R}_Y))
&\leq \gamma d_H(R_Y,\tilde{R}_Y).
\end{align}
\item 
If $P_{XY}(x,y)>0$ for all $x\in \mathcal{X},y\in \mathcal{Y}:$ 
For all $Q_X,\tilde{Q}_X\in \mathcal{P}(\mathcal{X}),R_Y,\tilde{R}_Y\in \mathcal{P}(\mathcal{Y})$
\begin{align}
d_H(\mathcal{N}_{X\rightarrow Y}(Q_X),\mathcal{N}_{X\rightarrow Y}(\tilde{Q}_X))
&\leq \gamma\kappa d_H(Q_X,\tilde{Q}_X),
\\
d_H(\mathcal{N}_{Y\rightarrow X}(R_Y),\mathcal{N}_{Y\rightarrow X}(\tilde{R}_Y))
&\leq \gamma\kappa d_H(R_Y,\tilde{R}_Y),
\end{align}
where $\kappa\coloneqq \tanh(\delta /4)\in [0,1)$ and 
\begin{align}
\delta \coloneqq \log \max_{\substack{x,x'\in \mathcal{X},\\ y,y'\in \mathcal{Y}}} \frac{P_{XY}(x,y)^\alpha P_{XY}(x',y')^\alpha}{P_{XY}(x',y)^\alpha P_{XY}(x,y')^\alpha}.
\end{align}
\end{enumerate}
\end{thm}
\begin{rem}[Subdivision of Theorem~\ref{thm:cc-linear}]
Note that the assertion in~(a) holds in general and guarantees a contraction coefficient of $\gamma \in [0, 1)$. 
In contrast,~(b) only holds under the assumption that $P_{XY}$ is strictly positive, but guarantees an improved contraction coefficient of $\gamma\kappa \in [0,\gamma) \subseteq [0, 1)$. 
This possibility of an improved coefficient will also appear analogously in subsequent results.
\end{rem}

The following result is a corollary of Theorem~\ref{thm:cc-linear}. 
It implies linear convergence of the PMFs obtained by alternating minimization with respect to Hilbert's projective metric. 
More specifically, (a) implies that the corresponding error rate is $\mathcal{O}(\lvert 1-\frac{1}{\alpha}\rvert^{2n})$.
\begin{cor}[Linear convergence of PMFs]\label{cor:cc-linear1}
\cite[Corollary~14]{tsai2024linearconvergencehilbertsprojective}
\sloppy
Let $\alpha\in [\frac{1}{2},1)\cup (1,\infty)$, 
$P_{XY}\in \mathcal{P}(\mathcal{X}\times \mathcal{Y})$, and 
let $(\hat{Q}_X,\hat{R}_Y)\in \argmin_{(Q_X,R_Y)\in \mathcal{P}(\mathcal{X})\times \mathcal{P}(\mathcal{Y})} D_\alpha (P_{XY}\| Q_XR_Y)$. 
Let $Q_X^{(0)}\in \mathcal{P}_{\sim P_X}(\mathcal{X})$ and let $R_Y^{(n)},Q_X^{(n+1)}$ be given by Definition~\ref{def:am-classical} for all $n\in \mathbb{N}$. 
Let $\gamma \coloneqq \lvert 1-\frac{1}{\alpha}\rvert$. 
\begin{enumerate}[label=(\alph*)]
\item 
If $\alpha >\frac{1}{2}$: 
For all $n\in \mathbb{N}$
\begin{align}
d_H(Q_X^{(n)},\hat{Q}_X) &\leq \gamma^{2n} d_H(Q_X^{(0)},\hat{Q}_X),
\\
d_H(R_Y^{(n)},\hat{R}_Y) &\leq \gamma^{2n+1} d_H(Q_X^{(0)},\hat{Q}_X).
\end{align}
\item 
If $P_{XY}(x,y)>0$ for all $x\in \mathcal{X},y\in \mathcal{Y}:$ 
For all $n\in \mathbb{N}$
\begin{align}
d_H(Q_X^{(n)},\hat{Q}_X) &\leq (\gamma\kappa)^{2n} d_H(Q_X^{(0)},\hat{Q}_X),
\\
d_H(R_Y^{(n)},\hat{R}_Y) &\leq (\gamma\kappa)^{2n+1} d_H(Q_X^{(0)},\hat{Q}_X)
\end{align}
where $\kappa$ is defined as in Theorem~\ref{thm:cc-linear}~(b).
\end{enumerate}
\end{cor}

The following result follows from Corollary~\ref{cor:cc-linear1}, as shown in Appendix~\ref{app:cc-linear2}. 
It implies linear convergence of the objective function values obtained by alternating minimization.
More specifically, (a) implies that 
$\lvert x_{n}-I_\alpha^{\downarrow\downarrow}(X:Y)_P \rvert$ is 
$\mathcal{O}(\lvert 1-\frac{1}{\alpha}\rvert^{2n})$ for any $\alpha\in (\frac{1}{2},\infty)$.
\begin{cor}[Linear convergence of objective function values]\label{cor:cc-linear2}
\sloppy
Let $\alpha\in [\frac{1}{2},1)\cup (1,\infty)$, 
$P_{XY}\in \mathcal{P}(\mathcal{X}\times \mathcal{Y})$, and 
let $(\hat{Q}_X,\hat{R}_Y)\in \argmin_{(Q_X,R_Y)\in \mathcal{P}(\mathcal{X})\times \mathcal{P}(\mathcal{Y})} D_\alpha (P_{XY}\| Q_XR_Y)$. 
Let $Q_X^{(0)}\in \mathcal{P}_{\sim P_X}(\mathcal{X})$ and let $(x_n)_{n\in \mathbb{N}}$ be given by Definition~\ref{def:am-classical}. 
Let $\gamma \coloneqq \lvert 1-\frac{1}{\alpha}\rvert$. 
\begin{enumerate}[label=(\alph*)]
\item 
If $\alpha >\frac{1}{2}$: 
For all $n\in \mathbb{N}$
\begin{align}
\lvert x_{n}-I_\alpha^{\downarrow\downarrow}(X:Y)_P \rvert 
&\leq \frac{1}{\lvert \alpha-1\rvert} \left(
\exp\left(\lvert \alpha-1\rvert (1+\gamma)\gamma^{2n} d_H(Q_X^{(0)},\hat{Q}_X) \right)
-1\right).
\end{align}
\item 
If $P_{XY}(x,y)>0$ for all $x\in \mathcal{X},y\in \mathcal{Y}:$ 
For all $n\in \mathbb{N}$
\begin{align}
\lvert x_{n}-I_\alpha^{\downarrow\downarrow}(X:Y)_P \rvert 
&\leq \frac{1}{\lvert \alpha-1\rvert} \left(
\exp\left(\lvert \alpha-1\rvert (1+\gamma\kappa)(\gamma\kappa)^{2n} d_H(Q_X^{(0)},\hat{Q}_X) \right)
-1\right)
\end{align}
where $\kappa$ is defined as in Theorem~\ref{thm:cc-linear}~(b).
\end{enumerate}
\end{cor}

\subsection{Proof of Corollary~\ref{cor:cc-linear2}}\label{app:cc-linear2}
\begin{lem}\label{lem:cc-m1}
Let $\alpha\in (0,1)\cup (1,\infty)$ and let $Q,R\in \mathcal{P}(\mathcal{X})$ be such that $Q\sim R$. 
Then, $M(Q^{1-\alpha}/R^{1-\alpha})\geq 1$.
\end{lem}
\begin{proof}
\emph{Case 1: $\alpha>1$.} 
\begin{align}
M(Q^{1-\alpha}/R^{1-\alpha})
&= \max_{x\in \supp(R)} \frac{Q(x)^{1-\alpha}}{R(x)^{1-\alpha}}
\label{eq:proof-cc-m1}\\
&\geq \sum_{x\in \supp(R)} R(x) \frac{Q(x)^{1-\alpha}}{R(x)^{1-\alpha}}
=\sum_{x\in \supp(R)} R(x)^{\alpha} Q(x)^{1-\alpha}
\\
&\geq \left(\sum_{x\in \supp(R)} (R(x)^\alpha )^{1/\alpha}\right)^\alpha
\left(\sum_{x\in \supp(R)} (Q(x)^{\alpha-1})^{1/(\alpha-1)}\right)^{1-\alpha}
\label{eq:proof-cc-m3}\\
&= \left(\sum_{x\in \supp(R)} R(x)\right)^\alpha
\left(\sum_{x\in \supp(R)} Q(x)\right)^{1-\alpha}
= 1
\label{eq:proof-cc-m4}
\end{align}
\eqref{eq:proof-cc-m1} follows from~\eqref{eq:m-xy-classical}. 
\eqref{eq:proof-cc-m3} follows from the reverse H\"older inequality~\cite[Lemma~3.1]{tomamichel2016quantum}.
The last equality in~\eqref{eq:proof-cc-m4} follows from $\sum_{x\in \mathcal{X}}R(x)=1=\sum_{x\in \mathcal{X}}Q(x)$.

\emph{Case 2: $\alpha<1$.} 
\begin{align}
M(Q^{1-\alpha}/R^{1-\alpha})
&= \max_{x\in \supp(R)} \frac{Q(x)^{1-\alpha}}{R(x)^{1-\alpha}}
\label{eq:proof-ccm1}\\
&\geq \sum_{x\in \supp(R)} Q(x) \frac{Q(x)^{1-\alpha}}{R(x)^{1-\alpha}}
=\sum_{x\in \supp(R)} Q(x)^{2-\alpha} R(x)^{\alpha-1}
\\
&\geq \left(\sum_{x\in \supp(R)} (Q(x)^{2-\alpha} )^{1/(2-\alpha)}\right)^{2-\alpha}
\left(\sum_{x\in \supp(R)} (R(x)^{1-\alpha})^{1/(1-\alpha)}\right)^{\alpha-1}
\label{eq:proof-ccm3}\\
&= \left(\sum_{x\in \supp(R)} Q(x)\right)^{2-\alpha}
\left(\sum_{x\in \supp(R)} R(x)\right)^{\alpha-1}
= 1
\label{eq:proof-ccm4}
\end{align}
\eqref{eq:proof-ccm1} follows from~\eqref{eq:m-xy-classical}. 
\eqref{eq:proof-ccm3} follows from the reverse H\"older inequality~\cite[Lemma~3.1]{tomamichel2016quantum}.
The last equality in~\eqref{eq:proof-ccm4} follows from $\sum_{x\in \mathcal{X}}R(x)=1=\sum_{x\in \mathcal{X}}Q(x)$. 
\end{proof}

\begin{proof}[Proof of Corollary~\ref{cor:cc-linear2}]
We will first prove~(a). 
Let $\hat{q}\coloneqq \exp((\alpha-1)I_\alpha^{\downarrow\downarrow}(X:Y)_P)$ 
and $q_n\coloneqq \exp((\alpha-1)x_n)$ for all $n\in \mathbb{N}$. 
Let $R_Y^{(n)},Q_X^{(n+1)}$ be given by Definition~\ref{def:am-classical} for all $n\in \mathbb{N}$. 

\emph{Case 1: $\alpha>1$.} 
For all $n\in \mathbb{N}$
\begin{align}
q_n-\hat{q}
&=\sum_{\substack{x\in \supp (P_X), \\ y\in \supp(P_Y)}}
P_{XY}(x,y)^\alpha \left((Q_X^{(n)}(x)R_Y^{(n)}(y) )^{1-\alpha}-(\hat{Q}_X(x)\hat{R}_Y(y) )^{1-\alpha}\right)
\\
&=\sum_{\substack{x\in \supp (P_X), \\ y\in \supp(P_Y)}}
P_{XY}(x,y)^\alpha (\hat{Q}_X(x) \hat{R}_Y(y) )^{1-\alpha} 
\left(\frac{(Q_X^{(n)}(x)R_Y^{(n)}(y) )^{1-\alpha}}{(\hat{Q}_X(x)\hat{R}_Y(y) )^{1-\alpha}}
-1\right)
\\
&\leq \hat{q} \left(
M((Q_X^{(n)} R_Y^{(n)})^{1-\alpha} / (\hat{Q}_X \hat{R}_Y)^{1-\alpha}) 
-1\right).
\label{eq:cc-2}
\end{align}

For all $n\in \mathbb{N}$
\begin{align}
\frac{q_n}{\hat{q}}
&\leq M( (Q_X^{(n)}R_Y^{(n)} )^{1-\alpha} / (\hat{Q}_X\hat{R}_Y )^{1-\alpha} )
\label{eq:cc-21}\\
&\leq M( (Q_X^{(n)}R_Y^{(n)} )^{1-\alpha} / (\hat{Q}_X\hat{R}_Y )^{1-\alpha} )
M( (\hat{Q}_X\hat{R}_Y )^{1-\alpha} / (Q_X^{(n)}R_Y^{(n)} )^{1-\alpha})
\label{eq:cc-22}
\\
&= \exp\left(
d_H( (Q_X^{(n)} R_Y^{(n)})^{1-\alpha}, (\hat{Q}_X \hat{R}_Y)^{1-\alpha} ) 
\right)
\label{eq:cc-23}\\
&\leq \exp\left(
\lvert 1-\alpha\rvert d_H( Q_X^{(n)} R_Y^{(n)}, \hat{Q}_X \hat{R}_Y ) 
\right)
\label{eq:cc-24}\\
&= \exp\left(
\lvert 1-\alpha\rvert (d_H( Q_X^{(n)}, \hat{Q}_X ) 
+d_H(R_Y^{(n)}, \hat{R}_Y ) )
\right)
\label{eq:cc-3}\\
&\leq \exp\left(\lvert 1-\alpha\rvert (1+\gamma)\gamma^{2n} d_H(Q_X^{(0)},\hat{Q}_X) \right).
\label{eq:cc-4}
\end{align}
\eqref{eq:cc-4} follows from Theorem~\ref{thm:cc-linear}. 
\eqref{eq:cc-21} follows from~\eqref{eq:cc2} because $\hat{q}=\exp((\alpha-1)I_\alpha^{\downarrow\downarrow}(X:Y)_P)>0$. 
\eqref{eq:cc-22} follows from Lemma~\ref{lem:cc-m1}.
\eqref{eq:cc-3} follows from the additivity of Hilbert’s projective metric with respect to the product of PMFs. 
\eqref{eq:cc-4} follows from Corollary~\ref{cor:cc-linear1}~(a). 
We conclude that for all $n\in \mathbb{N}$
\begin{align}
0\leq x_{n}-I_\alpha^{\downarrow\downarrow}(X:Y)_P 
&=\frac{1}{\alpha-1}\log\left(1+\frac{q_n}{\hat{q}}-1\right)
\\
&\leq \frac{1}{\alpha-1} \left(\frac{q_n}{\hat{q}}-1\right)
\label{eq:cc-6}\\
&\leq \frac{1}{\alpha-1 } \left(
\exp\left(\lvert\alpha-1\rvert (1+\gamma)\gamma^{2n} d_H(Q_X^{(0)},\hat{Q}_X) \right)
-1\right)
\label{eq:cc-7}
\\
&= \frac{1}{\lvert\alpha-1 \rvert} \left(
\exp\left(\lvert\alpha-1\rvert (1+\gamma)\gamma^{2n} d_H(Q_X^{(0)},\hat{Q}_X) \right)
-1\right).
\label{eq:cc-8}
\end{align}
\eqref{eq:cc-6} holds because 
$(q_n)_{n\in \mathbb{N}}$ is monotonically decreasing
and $\log(1+t)\leq t$ for all $t\in [0,\infty)$. 
\eqref{eq:cc-7} follows from~\eqref{eq:cc-4}. 
\eqref{eq:cc-8} holds since $\alpha>1$. 

\emph{Case 2: $\alpha<1$.} 
For all $n\in \mathbb{N}$
\begin{align}
\hat{q}-q_n
&=\sum_{\substack{x\in \supp (P_X), \\ y\in \supp(P_Y)}}
P_{XY}(x,y)^\alpha \left((\hat{Q}_X(x)\hat{R}_Y(y) )^{1-\alpha}-(Q_X^{(n)}(x)R_Y^{(n)}(y) )^{1-\alpha}\right)
\\
&=\sum_{\substack{x\in \supp (P_X), \\ y\in \supp(P_Y)}}
P_{XY}(x,y)^\alpha (Q_X^{(n)}(x)R_Y^{(n)}(y) )^{1-\alpha} 
\left(\frac{(\hat{Q}_X(x)\hat{R}_Y(y) )^{1-\alpha}}{(Q_X^{(n)}(x)R_Y^{(n)}(y) )^{1-\alpha}}-1\right)
\\
&\leq q_n \left(
M((\hat{Q}_X \hat{R}_Y)^{1-\alpha}/(Q_X^{(n)} R_Y^{(n)})^{1-\alpha}) 
-1\right).
\label{eq:cc2}
\end{align}

For all $n\in \mathbb{N}$
\begin{align}
\frac{\hat{q}}{q_n}
&\leq M((\hat{Q}_X \hat{R}_Y)^{1-\alpha}/(Q_X^{(n)} R_Y^{(n)})^{1-\alpha}) 
\label{eq:cc21}\\
&\leq M((\hat{Q}_X \hat{R}_Y)^{1-\alpha}/(Q_X^{(n)} R_Y^{(n)})^{1-\alpha}) 
M((Q_X^{(n)} R_Y^{(n)})^{1-\alpha}/(\hat{Q}_X \hat{R}_Y)^{1-\alpha}) 
\label{eq:cc22}\\
&= \exp\left(
d_H( (Q_X^{(n)} R_Y^{(n)})^{1-\alpha}, (\hat{Q}_X \hat{R}_Y)^{1-\alpha} ) 
\right)
\label{eq:cc23}\\
&\leq \exp\left( \lvert 1-\alpha\rvert 
d_H( Q_X^{(n)} R_Y^{(n)}, \hat{Q}_X \hat{R}_Y) 
\right)
\label{eq:cc24}\\
&=\exp\left( \lvert 1-\alpha\rvert
( d_H( Q_X^{(n)}, \hat{Q}_X) 
+ d_H( R_Y^{(n)}, \hat{R}_Y)  )
\right)
\label{eq:cc3}\\
&\leq \exp\left(\lvert 1-\alpha\rvert (1+\gamma)\gamma^{2n} d_H(Q_X^{(0)},\hat{Q}_X) \right).
\label{eq:cc4}
\end{align}
\eqref{eq:cc21} follows from~\eqref{eq:cc2} because $\hat{q}=\exp((\alpha-1)I_\alpha^{\downarrow\downarrow}(X:Y)_P)>0$. 
\eqref{eq:cc22} follows from Lemma~\ref{lem:cc-m1}.
\eqref{eq:cc3} follows from the additivity of Hilbert’s projective metric with respect to the product of PMFs. 
\eqref{eq:cc4} follows from Theorem~\ref{thm:cc-linear}. 
We conclude that for all $n\in \mathbb{N}$
\begin{align}
0\leq x_{n}-I_\alpha^{\downarrow\downarrow}(X:Y)_P 
&=\frac{1}{1-\alpha}\log\left(1+\frac{\hat{q}}{q_{n}}-1\right)
\\
&\leq \frac{1}{1-\alpha} \left(\frac{\hat{q}}{q_{n}}-1\right)
\label{eq:cc6}\\
&\leq \frac{1}{1-\alpha} \left(
\exp\left(\lvert 1-\alpha\rvert (1+\gamma)\gamma^{2n} d_H(Q_X^{(0)},\hat{Q}_X) \right)
-1\right)
\label{eq:cc7}
\\
&= \frac{1}{\lvert\alpha-1\rvert } \left(
\exp\left(\lvert\alpha-1\rvert (1+\gamma)\gamma^{2n} d_H(Q_X^{(0)},\hat{Q}_X) \right)
-1\right).
\label{eq:cc8}
\end{align}
\eqref{eq:cc6} holds because 
$(q_n)_{n\in \mathbb{N}}$ is monotonically increasing 
and $\log(1+t)\leq t$ for all $t\in [0,\infty)$. 
\eqref{eq:cc7} follows from~\eqref{eq:cc4}. 
\eqref{eq:cc8} holds since $\alpha<1$.
This completes the proof of~(a).

(b) can be proven analogously by applying Corollary~\ref{cor:cc-linear1}~(b) in place of Corollary~\ref{cor:cc-linear1}~(a). 
\end{proof}

\section{Properties of Hilbert's projective metric}\label{app:hilbert}
Let $V,W$ be real vector spaces and let $C\subseteq V,K\subseteq W$ be cones. 
A map $f:C\rightarrow K$ is \emph{order-preserving} 
if $x\leq_C y$ implies $f(x)\leq_K f(y)$ for all $x,y\in C$. 
A map $f:C\rightarrow K$ is \emph{order-reversing} 
if $x\leq_C y$ implies $f(y)\leq_K f(x)$ for all $x,y\in C$. 
A map $f:C\rightarrow K$ is \emph{homogeneous of degree $r\in \mathbb{R}$} 
if $f(\alpha x)=\alpha^rf(x)$ for all $x\in C,\alpha \in (0,\infty)$.

Hilbert's projective metric has the following properties.~\cite{lemmens2012nonlinear,lemmens2013birkhoffs}
\begin{enumerate}[label=(\alph*)]
\item \emph{Non-negativity:} 
$d_H(x,y)\in [0,\infty]$ for all $x,y\in C$.
\item \emph{Symmetry:} 
$d_H(x,y)=d_H(y,x)$ for all $x,y\in C$.
\item \emph{Scale invariance:} 
$d_H(\alpha x,\beta y)=d_H(x,y)$ for all $\alpha,\beta \in (0,\infty),x,y\in C$.
\item \emph{Projective definiteness:} If $C$ is a closed cone in a Banach space $(V,\lVert\cdot \rVert)$, then for all $x,y\in C$: $d_H(x,y)=0$ iff there exists $\alpha\in (0,\infty)$ such that $x=\alpha y$.
\item \emph{Homogeneous, order-preserving maps:}~\cite[Corollary~2.1.4]{lemmens2012nonlinear} 
Suppose $C\subseteq V$ and $K\subseteq W$ are closed cones.  
Let $f:C\rightarrow K$ be a map that is homogeneous of degree $r\in (0,\infty)$ and order-preserving. Then,  
\begin{align}
d_H(f(x),f(y))\leq r d_H(x,y)
\end{align}
for all $x,y\in C$ such that $x\sim_C y$.
\item \emph{Homogeneous, order-reversing maps:}~\cite[Corollary~2.1.5]{lemmens2012nonlinear} 
Suppose $C\subseteq V$ and $K\subseteq W$ are closed cones.  
Let $f:C\rightarrow K$ be a map that is homogeneous of degree $r\in (-\infty,0)$ and order-reversing. Then, 
\begin{align}
d_H(f(x),f(y))\leq \lvert r\rvert d_H(x,y)
\end{align} 
for all $x,y\in C$ such that $x\sim_C y$.
\item \emph{Birkhoff-Hopf theorem:}~\cite{birkhoff1957extensions,hopf1963inequality,lemmens2012nonlinear,lemmens2013birkhoffs} 
Let $L:V\rightarrow W$ be a linear map such that $L(C)\subseteq K$. 
The \emph{Birkhoff contraction ratio of $L$} is defined as 
\begin{align}
\kappa (L)\coloneqq \inf \{\alpha\in [0,\infty): d_H(L(x),L(y))\leq \alpha d_H(x,y) \text{ for all } x,y\in C \text{ such that }x\sim_C y\},
\end{align}
and the \emph{projective diameter of $L$} is defined as
\begin{align}
\delta (L)\coloneqq \sup \{d_H(L(x),L(y)):x,y\in C \text{ such that }L(x)\sim_K L(y) \}.
\end{align}

The Birkhoff-Hopf theorem then asserts that 
\begin{align}
\kappa (L)=\tanh \left(\frac{\delta (L)}{4}\right)
\end{align}
where $\tanh(\infty)= 1$.
\end{enumerate}

\section{Proofs for Section~\ref{ssec:linear}}
\subsection{Proof of Theorem~\ref{thm:linear}}\label{app:linear}
\begin{lem}\label{lem:quantum-r}
All of the following hold.
\begin{enumerate}[label=(\alph*)]
\item Let $r\in [-1,1]$. 
Then, for all positive semidefinite $X,Y\in \mathcal{L}(A)$ such that $X\sim Y$
\begin{align}
d_H(X^r,Y^r)\leq \lvert r\rvert d_H(X,Y).
\end{align}
\item Let $Z_{AB}\in \mathcal{L}(AB)$ be positive semidefinite. 
Let $L:\{X_A\in \mathcal{L}(A):X_A^\dagger=X_A\}\rightarrow \{Y_B\in \mathcal{L}(B):Y_B^\dagger=Y_B\},X_A\mapsto \tr_A[Z_{AB}X_A]$. 
Then, for all positive semidefinite $X_A,Y_A\in \mathcal{L}(A)$ such that $X_A\sim Y_A$
\begin{align}\label{eq:lem-z}
d_H(L(X_A),L(Y_A))\leq \kappa (L) d_H(X_A,Y_A).
\end{align}
Moreover, 
\begin{align}\label{eq:delta-l-sup}
\delta (L)
&=\sup_{\substack{\ket{\phi},\ket{\psi}\in A:\\ \brak{\phi}=1,\brak{\psi}=1}} 
d_H(L(\proj{\phi}),L(\proj{\psi})).
\end{align}
\end{enumerate}
\end{lem}

\begin{proof}[Proof of (a)]
Let $f(X)\coloneqq X^r$ for all positive semidefinite $X\in \mathcal{L}(A)$. 
$f$ is homogeneous of degree $r$.

\emph{Case 1: $r\in (0,1]$.} 
Then $f$ is operator monotone (i.e., order-preserving). 
The assertion then follows from Appendix~\ref{app:hilbert}~(e).

\emph{Case 2: $r\in [-1,0)$.} 
Then $f$ is operator anti-monotone (i.e., order-reversing). 
The assertion then follows from Appendix~\ref{app:hilbert}~(f).

\emph{Case 3: $r=0$.} 
Then $X^r=Y^r$ because $X\sim Y$. Hence, $d_H(X^r,Y^r)=d_H(X^r,X^r)=0$ by the projective definiteness of Hilbert's projective metric, see Appendix~\ref{app:hilbert}~(d).
\end{proof}
\begin{proof}[Proof of (b)]

$L$ is an $\mathbb{R}$-linear map (i.e., homogeneous of degree $1$) on the vector space of self-adjoint operators on $A$, 
and $L$ is operator monotone (i.e., order-preserving) on the cone of positive semidefinite operators on $A$. 
The assertion in~\eqref{eq:lem-z} then follows from Appendix~\ref{app:hilbert}~(g).

It remains to prove~\eqref{eq:delta-l-sup}. 
Let $S\coloneqq \{\proj{\phi}: \ket{\phi}\in A,\brak{\phi}=1\}$. Then,
\begin{align}
\delta (L)
&=\sup_{\substack{X,Y\in \mathcal{L}(A)\setminus \{0\}: \\ X\geq 0,Y\geq 0}} d_H(L(X),L(Y))
\label{eq:proof-lem-1}\\
&=\sup_{X,Y\in \mathcal{S}(A)} d_H(L(X),L(Y))
\label{eq:proof-lem-2}\\
&=\sup_{X,Y\in S} d_H(L(X),L(Y))
\label{eq:proof-lem-3}\\
&=\sup_{\substack{\ket{\phi},\ket{\psi}\in A:\\ \brak{\phi}=1,\brak{\psi}=1}} 
d_H(L(\proj{\phi}),L(\proj{\psi})).
\end{align}
\eqref{eq:proof-lem-2} follows from the scale invariance of Hilbert's projective metric, see Appendix~\ref{app:hilbert}~(c). 
\eqref{eq:proof-lem-3} holds because the convex hull of $S$ is $\mathcal{S}(A)$~\cite[Lemma~A.3.4]{lemmens2012nonlinear}.
\end{proof}

\begin{proof}[Proof of Theorem~\ref{thm:linear}]
Let $\sigma_A,\tilde{\sigma}_A\in \mathcal{S}_{\sim\rho_A}(A)$. 
Let us define the linear map 
$L:\{X_A\in \mathcal{L}(A):X_A^\dagger=X_A\}\rightarrow \{Y_B\in \mathcal{L}(B):Y_B^\dagger=Y_B\},X_A\mapsto \tr_A[\rho_{AB}^\alpha X_A]$. 
\begin{align}
d_H(\mathcal{N}_{A\rightarrow B}(\sigma_A), \mathcal{N}_{A\rightarrow B}(\tilde{\sigma}_A))
&=d_H((\tr_A[\rho_{AB}^\alpha \sigma_A^{1-\alpha}])^{\frac{1}{\alpha}},(\tr_A[\rho_{AB}^\alpha \tilde{\sigma}_A^{1-\alpha}])^{\frac{1}{\alpha}} )
\label{eq:proof-thml1}\\
&\leq \frac{1}{\alpha} d_H(\tr_A[\rho_{AB}^\alpha \sigma_A^{1-\alpha}],\tr_A[\rho_{AB}^\alpha \tilde{\sigma}_A^{1-\alpha}])
\label{eq:proof-thml2}
\\
&= \frac{1}{\alpha} d_H(L(\sigma_A^{1-\alpha}),L(\tilde{\sigma}_A^{1-\alpha}))
\\
&\leq \frac{1}{\alpha} \kappa (L) d_H(\sigma_A^{1-\alpha},\tilde{\sigma}_A^{1-\alpha})
\label{eq:proof-thml3}
\\
&\leq \frac{\lvert 1-\alpha\rvert}{\alpha}\kappa (L) d_H(\sigma_A,\tilde{\sigma}_A)
=\gamma \kappa (L) d_H(\sigma_A,\tilde{\sigma}_A)
\label{eq:proof-thml4}
\\
&\leq \gamma d_H(\sigma_A,\tilde{\sigma}_A)
\label{eq:proof-thml5}
\end{align}
\eqref{eq:proof-thml1} follows from Definition~\ref{def:am-quantum} and the scale invariance of Hilbert's projective metric, see Appendix~\ref{app:hilbert}~(c). 
\eqref{eq:proof-thml2} follows from Lemma~\ref{lem:quantum-r}~(a) because $\alpha\in (1,2]\subseteq [1,\infty)$ and $\rho_A\ll \sigma_A$, $\rho_A\ll \tilde{\sigma}_A$. 
\eqref{eq:proof-thml3} follows from Lemma~\ref{lem:quantum-r}~(b) because $\sigma_A\sim \tilde{\sigma}_A$. 
The inequality in~\eqref{eq:proof-thml4} follows from Lemma~\ref{lem:quantum-r}~(a) because $\alpha\in (1,2]\subseteq [0,2]$. 
\eqref{eq:proof-thml4} holds because $\kappa(L)\in [0,1]$, see Appendix~\ref{app:hilbert}~(g). 
This completes the proof of~\eqref{eq:thm-linear-sigma1}. 
The proof of~\eqref{eq:thm-linear-tau1} is analogous.
This completes the proof of~(a).

The assertion in~(b) follows from the proof of~(a), see~\eqref{eq:proof-thml4}, and Lemma~\ref{lem:quantum-r}~(b), see~\eqref{eq:delta-l-sup}.
\end{proof}

\subsection{Proof of Corollary~\ref{cor:linear1}}\label{app:linear1}
\begin{proof}
By the fixed-point property of minimizers for the doubly minimized PRMI~\cite{burri2024properties}, 
\begin{align}
\hat{\sigma}_A&=\mathcal{N}_{B\rightarrow A}\circ \mathcal{N}_{A\rightarrow B}(\hat{\sigma}_A)
\in \mathcal{S}_{\ll \rho_A}(A),
\label{eq:fixed-sigma}\\
\hat{\tau}_B&=\mathcal{N}_{A\rightarrow B}(\hat{\sigma}_A)
\in \mathcal{S}_{\ll \rho_B}(B).
\label{eq:fixed-tau}
\end{align}
Since $\alpha\geq 1$, $\hat{\sigma}_A\in \mathcal{S}_{\rho_A\ll}(A)$ and $\hat{\tau}_B\in \mathcal{S}_{\rho_B\ll}(B)$~\cite{burri2024properties}. 
We conclude that $\hat{\sigma}_A\sim \rho_A$ and $\hat{\tau}_B\sim \rho_B$.

For all $n\in \mathbb{N}$,
\begin{align}
d_H(\sigma_A^{(n+1)},\hat{\sigma}_A) 
&=d_H(\mathcal{N}_{B\rightarrow A}\circ \mathcal{N}_{A\rightarrow B}(\sigma_A^{(n)}),\mathcal{N}_{B\rightarrow A}\circ \mathcal{N}_{A\rightarrow B}(\hat{\sigma}_A))
\label{eq:proof-cor-1}\\
&\leq \gamma d_H(\mathcal{N}_{A\rightarrow B}(\sigma_A^{(n)}), \mathcal{N}_{A\rightarrow B}(\hat{\sigma}_A))
\label{eq:proof-cor-2}\\
&\leq \gamma^2 d_H(\sigma_A^{(n)},\hat{\sigma}_A).
\label{eq:proof-cor-3}
\end{align}
\eqref{eq:proof-cor-1} follows from Definition~\ref{def:am-quantum} and~\eqref{eq:fixed-sigma}.
\eqref{eq:proof-cor-2} and~\eqref{eq:proof-cor-3} follow from Theorem~\ref{thm:linear}~(a).

By recursion, we conclude that for all $n\in \mathbb{N}$
\begin{align}
d_H(\sigma_A^{(n)},\hat{\sigma}_A) \leq \gamma^{2n} d_H(\sigma_A^{(0)},\hat{\sigma}_A).
\label{eq:proof-cor-31}
\end{align}

For all $n\in \mathbb{N}$,
\begin{align}
d_H(\tau_B^{(n)},\hat{\tau}_B) 
&=d_H(\mathcal{N}_{A\rightarrow B}(\sigma_A^{(n)}),\mathcal{N}_{A\rightarrow B}(\hat{\sigma}_A))
\label{eq:proof-cor-4}\\
&\leq \gamma d_H(\sigma_A^{(n)},\hat{\sigma}_A)
\label{eq:proof-cor-5}\\
&\leq \gamma^{2n+1} d_H(\sigma_A^{(0)},\hat{\sigma}_A) .
\label{eq:proof-cor-6}
\end{align}
\eqref{eq:proof-cor-4} follows from Definition~\ref{def:am-quantum} and~\eqref{eq:fixed-tau}. 
\eqref{eq:proof-cor-5} follows from Theorem~\ref{thm:linear}~(a).
\eqref{eq:proof-cor-6} follows from~\eqref{eq:proof-cor-31}.
This completes the proof of~(a). 

(b) can be proven analogously by applying Theorem~\ref{thm:linear}~(b) in place of Theorem~\ref{thm:linear}~(a).
\end{proof}

\subsection{Proof of Corollary~\ref{cor:linear2}}\label{app:linear2}
\begin{lem}\label{lem:m1}
Let $\alpha\in (1,\infty)$ and let $\sigma,\tau\in \mathcal{S}(A)$ be such that $\sigma\sim \tau$. 
Then, $M(\sigma^{1-\alpha}/\tau^{1-\alpha})\geq 1$.
\end{lem}
\begin{proof}
\begin{align}
M(\sigma^{1-\alpha}/\tau^{1-\alpha})
&=\lVert \tau^{\frac{\alpha-1}{2}}\sigma^{1-\alpha} \tau^{\frac{\alpha-1}{2}}\rVert_\infty
\label{eq:proof-m1}\\
&=\max_{\ket{\phi}\in A: \brak{\phi}=1} 
\tr[\tau^{\frac{\alpha-1}{2}}\sigma^{1-\alpha} \tau^{\frac{\alpha-1}{2}} \proj{\phi}]
\\
&=\max_{\rho\in \mathcal{S}(A)} 
\tr[\tau^{\frac{\alpha-1}{2}}\sigma^{1-\alpha} \tau^{\frac{\alpha-1}{2}} \rho]
\\
&\geq 
\tr[\tau^{\frac{\alpha-1}{2}}\sigma^{1-\alpha} \tau^{\frac{\alpha-1}{2}} \tau]
=\tr[\tau^{\alpha}\sigma^{1-\alpha}]
\\
&\geq \lVert \tau^\alpha \rVert_{1/\alpha} \cdot \lVert \sigma^{\alpha-1}\rVert_{1/(\alpha -1)}^{-1}
\label{eq:proof-m3}\\
&=\tr[\tau]^\alpha \cdot (\tr[\sigma]^{\alpha-1})^{-1}
=1
\label{eq:proof-m4}
\end{align}
\eqref{eq:proof-m1} follows from~\eqref{eq:m-xy-quantum}. 
\eqref{eq:proof-m3} follows from the reverse H\"older inequality~\cite[Lemma~3.1]{tomamichel2016quantum}.
The last equality in~\eqref{eq:proof-m4} follows from $\tr[\tau]=1=\tr[\sigma]$.
\end{proof}

\begin{proof}[Proof of Corollary~\ref{cor:linear2}]
Let $\hat{q}\coloneqq Q_\alpha (\rho_{AB}\| \hat{\sigma}_A\otimes \hat{\tau}_B)$. 
Let $\tau_B^{(n)},\sigma_A^{(n+1)}$ be given by Definition~\ref{def:am-quantum} for all $n\in \mathbb{N}$. 
Let $q_n\coloneqq Q_\alpha (\rho_{AB}\| \sigma_A^{(n)}\otimes \tau_B^{(n)})$ for all $n\in \mathbb{N}$. 

For all $n\in \mathbb{N}$
\begin{align}
q_n-\hat{q}
&= \tr[\rho_{AB}^\alpha ( (\sigma_A^{(n)}\otimes \tau_B^{(n)})^{1-\alpha} - (\hat{\sigma}_A\otimes \hat{\tau}_B)^{1-\alpha} )]
\\
&=\tr[(\hat{\sigma}_A\otimes \hat{\tau}_B)^{\frac{1-\alpha}{2}}\rho_{AB}^\alpha 
(\hat{\sigma}_A\otimes \hat{\tau}_B)^{\frac{1-\alpha}{2}} 
\nonumber\\
&\qquad \left( (\hat{\sigma}_A\otimes \hat{\tau}_B)^{-\frac{1}{2}(1-\alpha)} 
(\sigma_A^{(n)}\otimes \tau_B^{(n)})^{1-\alpha} 
(\hat{\sigma}_A\otimes \hat{\tau}_B)^{-\frac{1}{2}(1-\alpha)} 
-1\right)]
\label{eq:q1}\\
&\leq \tr[\rho_{AB}^\alpha 
(\hat{\sigma}_A\otimes \hat{\tau}_B)^{1-\alpha} ]
\left(M((\sigma_A^{(n)}\otimes \tau_B^{(n)})^{1-\alpha} / (\hat{\sigma}_A\otimes \hat{\tau}_B)^{1-\alpha}) -1 \right)
\label{eq:q2}\\
&= \hat{q} \left(M((\sigma_A^{(n)}\otimes \tau_B^{(n)})^{1-\alpha} / (\hat{\sigma}_A\otimes \hat{\tau}_B)^{1-\alpha}) -1 \right).
\label{eq:q3}
\end{align}
\eqref{eq:q1} holds because $\sigma_A^{(n)}\sim \rho_A\sim \hat{\sigma}_A$ and $\tau_B^{(n)}\sim \rho_B\sim \hat{\tau}_B$. 
\eqref{eq:q2} follows from~\eqref{eq:m-xy-quantum}. 

For all $n\in \mathbb{N}$
\begin{align}
\frac{q_n}{\hat{q}}
&\leq M((\sigma_A^{(n)}\otimes \tau_B^{(n)})^{1-\alpha} / (\hat{\sigma}_A\otimes \hat{\tau}_B)^{1-\alpha})
\label{eq:q30}\\
&\leq M((\sigma_A^{(n)}\otimes \tau_B^{(n)})^{1-\alpha} / (\hat{\sigma}_A\otimes \hat{\tau}_B)^{1-\alpha})
M( (\hat{\sigma}_A\otimes \hat{\tau}_B)^{1-\alpha} / (\sigma_A^{(n)}\otimes \tau_B^{(n)})^{1-\alpha} )
\label{eq:q4}\\
&= \exp \left(d_H( 
(\sigma_A^{(n)}\otimes \tau_B^{(n)})^{1-\alpha} , (\hat{\sigma}_A\otimes \hat{\tau}_B)^{1-\alpha}
)\right)
\\
&\leq \exp \left(\lvert 1-\alpha\rvert
d_H(\sigma_A^{(n)}\otimes \tau_B^{(n)},\hat{\sigma}_A\otimes \hat{\tau}_B) \right)
\label{eq:q40}\\
&= \exp \left(\lvert 1-\alpha\rvert
(d_H(\sigma_A^{(n)},\hat{\sigma}_A)
+d_H(\tau_B^{(n)},\hat{\tau}_B))\right)
\label{eq:q41}\\
&\leq \exp \left(\lvert 1-\alpha\rvert
(1+\gamma)\gamma^{2n}d_H(\sigma_A^{(0)},\hat{\sigma}_A) \right).
\label{eq:q42}
\end{align}
\eqref{eq:q30} follows from~\eqref{eq:q3} because $\hat{q}=\exp((\alpha-1)I_\alpha^{\downarrow\downarrow}(A:B)_\rho)>0$. 
\eqref{eq:q4} follows from Lemma~\ref{lem:m1}. 
\eqref{eq:q40} follows from Lemma~\ref{lem:quantum-r}~(a) because $\alpha\in [1,2]\subseteq [0,2]$. 
\eqref{eq:q41} follows from the additivity of Hilbert's projective metric with respect to the tensor product~\cite{reeb2011hilbert}. 
\eqref{eq:q42} follows from Corollary~\ref{cor:linear1}~(a). 
We conclude that for all $n\in \mathbb{N}$
\begin{align}
0\leq x_{n}-I_\alpha^{\downarrow\downarrow}(A:B)_\rho 
&=\frac{1}{\alpha-1}\log\left(1+\frac{q_n}{\hat{q}}-1\right)
\\
&\leq \frac{1}{\alpha-1} \left(\frac{q_n}{\hat{q}}-1\right)
\label{eq:q6}\\
&\leq \frac{1}{\alpha-1 } \left(
\exp\left(\lvert 1-\alpha\rvert (1+\gamma)\gamma^{2n} d_H(\sigma_A^{(0)},\hat{\sigma}_A) \right)
-1\right)
\label{eq:q7}
\\
&= \frac{1}{\lvert \alpha-1 \rvert} \left(
\exp\left(\lvert \alpha-1\rvert (1+\gamma)\gamma^{2n} d_H(\sigma_A^{(0)},\hat{\sigma}_A) \right)
-1\right).
\label{eq:q8}
\end{align}
\eqref{eq:q6} holds because 
$(q_n)_{n\in \mathbb{N}}$ is monotonically decreasing
and $\log(1+t)\leq t$ for all $t\in [0,\infty)$. 
\eqref{eq:q7} follows from~\eqref{eq:q42}. 
\eqref{eq:q8} holds since $\alpha>1$. 
This completes the proof of~(a).

(b) can be proven analogously by applying Corollary~\ref{cor:linear1}~(b) in place of Corollary~\ref{cor:linear1}~(a). 
\end{proof}

\subsection{Proposition for Remark~\ref{rem:dh-sigma}}\label{app:dh-sigma}
\begin{lem}[Upper bound on Hilbert's metric for states]\label{lem:dh-minspec}
Let $\sigma,\tau\in \mathcal{S}(A)$ be such that $\sigma\sim \tau$. Then,
\begin{align}
d_H(\sigma,\tau)
\leq -2\log \min\{\min (\spec(\sigma)\setminus \{0\}),\min(\spec(\tau)\setminus \{0\})\}.
\end{align}
\end{lem}
\begin{proof}
\begin{align}
M(\sigma/\tau) &=\lVert \tau^{-\frac{1}{2}}\sigma\tau^{-\frac{1}{2}}\rVert_\infty 
\leq \lVert \tau^{-1}\rVert_\infty = \frac{1}{\min (\spec(\tau)\setminus \{0\})}
\\
M(\tau/\sigma) &=\lVert \sigma^{-\frac{1}{2}}\tau\sigma^{-\frac{1}{2}}\rVert_\infty 
\leq \lVert \sigma^{-1}\rVert_\infty = \frac{1}{\min (\spec(\sigma)\setminus \{0\})}
\end{align}
Therefore,
\begin{align}
d_H(\sigma,\tau)
\leq 2 \log \max\{M(\sigma/\tau),M(\tau/\sigma)\}
&\leq -2\log \min\{\min (\spec(\sigma)\setminus \{0\}),\min(\spec(\tau)\setminus \{0\})\}.
\end{align}
\end{proof}

\begin{lem}[Lower bound on spectrum of states from alternating minimization]\label{lem:spectrum-geq1}
Let $\alpha\in (1,\infty)$, 
$\rho_{AB}\in \mathcal{S}(AB)$, 
$\sigma_A^{(0)}\in \mathcal{S}_{\sim \rho_A}(A)$. 
Let $\tau_B^{(n)}$, $\sigma_A^{(n+1)}$ for $n\in \mathbb{N}$ be given by Definition~\ref{def:am-quantum}. 
Let us define the following real numbers. 
\begin{align}
\lambda_A&\coloneqq \min(\spec(\tr_B[\rho_{AB}^\alpha])\setminus\{0\})
\\
\lambda_B&\coloneqq \min(\spec(\tr_A[\rho_{AB}^\alpha])\setminus\{0\})
\\
\lambda_{A,n}
&\coloneqq \min (\spec(\sigma_A^{(n)})\setminus \{0\})
&&\forall n\in \mathbb{N}
\\
\lambda_{B,n}
&\coloneqq \min (\spec(\tau_B^{(n)})\setminus \{0\}) 
&&\forall n\in \mathbb{N}
\\
q_0 
&\coloneqq Q_\alpha(\rho_{AB}\| \sigma_A^{(0)}\otimes \tau_B^{(0)})
\\
c_A
&\coloneqq (\lambda_A/q_0)^{\frac{1}{\alpha}}
\\
c_B
&\coloneqq (\lambda_B/q_0)^{\frac{1}{\alpha}}
\end{align}
Then 
$\lambda_{A,n+1}\geq c_{A}>0$ and 
$\lambda_{B,n}\geq c_{B}>0$
for all $n\in \mathbb{N}$.
\end{lem}
\begin{proof}
For all $n \in \mathbb{N}$, let
\begin{subequations}\label{eq:def-12}
\begin{align}
x_{n}&\coloneqq D_\alpha (\rho_{AB}\| \sigma_A^{(n)}\otimes \tau_B^{(n)}), \qquad 
q_n\coloneqq Q_\alpha (\rho_{AB}\| \sigma_A^{(n)}\otimes \tau_B^{(n)}),
\\
x_{n+1/2}&\coloneqq D_\alpha (\rho_{AB}\| \sigma_A^{(n+1)}\otimes \tau_B^{(n)}), \qquad 
q_{n+1/2}\coloneqq Q_\alpha (\rho_{AB}\| \sigma_A^{(n+1)}\otimes \tau_B^{(n)}).
\end{align}
\end{subequations}
The sequence $(x_n)_{n\in \frac{1}{2}\mathbb{N}}$ is monotonically decreasing by construction, see~\eqref{eq:tau-optimal} and~\eqref{eq:inf-tau}. 
$x_n\in [0,\infty)$ for all $n\in \frac{1}{2}\mathbb{N}$ due to the non-negativity of the Petz divergence. 
Hence, $q_n=\exp((\alpha-1)x_n)\in [1,\infty)$ for all $n\in \frac{1}{2}\mathbb{N}$, and 
$(q_n)_{n\in \frac{1}{2}\mathbb{N}}$ is monotonically decreasing. 

Let $\mathcal{N}_{A\rightarrow B}$, $\mathcal{N}_{B\rightarrow A}$ be defined as in~\eqref{eq:def-N}.
By~\eqref{eq:inf-tau}, we have for all $n\in \mathbb{N}$ 
\begin{align}
q_n&=\exp\left((\alpha-1)D_\alpha (\rho_{AB}\| \sigma_A^{(n)}\otimes \mathcal{N}_{A\rightarrow B}(\sigma_A^{(n)}))\right)
=\tr[(\tr_A[\rho_{AB}^\alpha \sigma_A^{(n)1-\alpha}])^{\frac{1}{\alpha}} ]^\alpha,
\label{eq:qn0}\\
q_{n+1/2}&=\exp\left((\alpha-1)D_\alpha (\rho_{AB}\| \mathcal{N}_{B\rightarrow A}(\tau_B^{(n)})\otimes \tau_B^{(n)} )\right)
=\tr[(\tr_B[\rho_{AB}^\alpha \tau_B^{(n)1-\alpha}])^{\frac{1}{\alpha}} ]^\alpha.
\label{eq:qn1}
\end{align} 

For all $n\in \mathbb{N}$
\begin{align}
q_0 \tau_B^{(n)\alpha}
&\geq q_n \tau_B^{(n)\alpha}
\label{eq:q0qn1}\\
&=q_{n}\left(\mathcal{N}_{A\rightarrow B}(\sigma_{A}^{(n)})\right)^\alpha
\stackrel{\eqref{eq:qn0}}{=} \tr_A[\rho_{AB}^\alpha \sigma_A^{(n)1-\alpha} ]
\\
&\geq \tr_A[\rho_{AB}^\alpha \sigma_A^{(n)0}] 
=\tr_A[\rho_{AB}^\alpha]
\geq \lambda_B\rho_B^0 =\lambda_B \tau_B^{(n)0}.
\end{align}
\eqref{eq:q0qn1} holds because $(q_n)_{n\in \frac{1}{2}\mathbb{N}}$ is monotonically decreasing. 
Therefore, for all $n\in \mathbb{N}$
\begin{align}
\lambda_{B,n}\geq \left(\frac{\lambda_B}{q_0}\right)^{\frac{1}{\alpha}}
=c_B.
\end{align}

For all $n\in \mathbb{N}$
\begin{align}
q_0 \sigma_A^{(n+1)\alpha}
&\geq q_{n+1/2} \sigma_A^{(n+1)\alpha}
\label{eq:q0qn2}\\
&=q_{n+1/2}\left(\mathcal{N}_{B\rightarrow A}(\tau_B^{(n)})\right)^\alpha
\stackrel{\eqref{eq:qn1}}{=} \tr_B[\rho_{AB}^\alpha \tau_B^{(n)1-\alpha} ]
\\
&\geq \tr_B[\rho_{AB}^\alpha \tau_B^{(n)0}] 
=\tr_B[\rho_{AB}^\alpha ]
\geq \lambda_A\rho_A^0 =\lambda_A \sigma_A^{(n+1)0}.
\end{align}
\eqref{eq:q0qn2} holds because $(q_n)_{n\in \frac{1}{2}\mathbb{N}}$ is monotonically decreasing. 
Therefore, for all $n\in \mathbb{N}$
\begin{align}
\lambda_{A,n+1}\geq \left(\frac{\lambda_A}{q_0}\right)^{\frac{1}{\alpha}}
=c_A.
\end{align}
\end{proof}

\begin{prop}[Upper bound independent of $\hat{\sigma}_A$]\label{prop:bound}
Let $\alpha\in (1,\infty)$, 
$\rho_{AB}\in \mathcal{S}(AB)$, 
$(\hat{\sigma}_A,\hat{\tau}_B)\in \argmin_{(\sigma_A,\tau_B)\in \mathcal{S}(A)\times \mathcal{S}(B)}D_\alpha (\rho_{AB}\| \sigma_A\otimes \tau_B)$. 
Let $\sigma_A^{(0)}\in \mathcal{S}_{\rho_A\ll}(A)$ 
and let $\mathcal{N}_{A\rightarrow B}$ be given by Definition~\ref{def:am-quantum}. 
Moreover, let us define the following objects.
\begin{subequations}\label{eq:def-linear-c0}
\begin{align}
\lambda_A&\coloneqq \min(\spec(\tr_B[\rho_{AB}^\alpha])\setminus\{0\})
\\
\tilde{\sigma}_A^{(0)}
&\coloneqq \frac{\rho_A^0 \sigma_A^{(0)} \rho_A^0}{\tr[\rho_A^0\sigma_A^{(0)}]}
\\
q_0 
&\coloneqq Q_\alpha(\rho_{AB}\| \tilde{\sigma}_A^{(0)}\otimes \mathcal{N}_{A\rightarrow B}(\tilde{\sigma}_A^{(0)}) )
\\
c_A
&\coloneqq (\lambda_A/q_0)^{\frac{1}{\alpha}}
\\
c_0
&\coloneqq -2\log \min\{\min ( \spec(\tilde{\sigma}_A^{(0)})\setminus \{0\} ),c_A\}
\end{align}
\end{subequations}
Then,
$d_H(\tilde{\sigma}_A^{(0)},\hat{\sigma}_A)
\leq c_0.$
\end{prop}
\begin{proof}
$\hat{\sigma}_A\in \mathcal{S}_{\sim \rho_A}(A)$ because $\alpha\in (1,\infty)$~\cite{burri2024properties}. 
Therefore, $\hat{\sigma}_A\sim \rho_A\sim \tilde{\sigma}_A^{(0)}$.
\begin{align}
d_H(\tilde{\sigma}_A^{(0)},\hat{\sigma}_A)
&\leq -2\log \min\{\min (\spec(\tilde{\sigma}_A^{(0)})\setminus \{0\}),\min(\spec(\hat{\sigma}_A)\setminus \{0\})\}
\label{eq:proof-prop1}\\
&\leq -2\log \min\{\min (\spec(\tilde{\sigma}_A^{(0)})\setminus \{0\}),c_A\} 
=c_0
\label{eq:proof-prop2}
\end{align}
\eqref{eq:proof-prop1} follows from Lemma~\ref{lem:dh-minspec}. 
\eqref{eq:proof-prop2} follows from Lemma~\ref{lem:spectrum-geq1}. 
\end{proof}

\section{Proofs for Section~\ref{ssec:sublinear}}\label{app:proof-frechet}
The proof of Theorem~\ref{thm:sublinear} invokes two lemmas established in advance.

\subsection{Lemmas for Theorem~\ref{thm:sublinear}}\label{app:lemmas}
\begin{lem}[Lower bound on spectrum of states from alternating minimization]\label{lem:spectrum-leq1}
Let $\alpha\in (\frac{1}{2},1)$, 
$\rho_{AB}\in \mathcal{S}(AB)$, 
$\sigma_A^{(0)}\in \mathcal{S}_{\sim \rho_A}(A)$. 
Let $\tau_B^{(n)}$, $\sigma_A^{(n+1)}$ for $n\in \mathbb{N}$ be given by Definition~\ref{def:am-quantum}. 
Let us define the following real numbers. 
\begin{align}
\lambda_A&\coloneqq \min(\spec(\tr_B[\rho_{AB}^\alpha])\setminus\{0\})
\\
\lambda_B&\coloneqq \min(\spec(\tr_A[\rho_{AB}^\alpha])\setminus\{0\})
\\
\lambda_{A,n}
&\coloneqq \min (\spec(\sigma_A^{(n)})\setminus \{0\})
&&\forall n\in \mathbb{N}
\label{eq:lambda-an}\\
\lambda_{B,n}
&\coloneqq \min (\spec(\tau_B^{(n)})\setminus \{0\}) 
&&\forall n\in \mathbb{N}
\label{eq:lambda-bn}\\
c_{A}&\coloneqq \min(1, \lambda_A^\frac{\alpha}{2\alpha-1} \lambda_B^{\frac{1-\alpha}{2\alpha-1}} ) \lambda_{A,0}
\\
c_{B}&\coloneqq \min(1, \lambda_A^\frac{1-\alpha}{2\alpha-1} \lambda_B^{\frac{(1-\alpha)^2}{(2\alpha-1)\alpha}} ) \lambda_B^{\frac{1}{\alpha}} \lambda_{A,0}^{\frac{1-\alpha}{\alpha}}
\label{eq:def-cb}
\end{align}
Then 
$\lambda_{A,n}\geq c_{A}>0$ and 
$\lambda_{B,n}\geq c_{B}>0$
for all $n\in \mathbb{N}$.
\end{lem}

\begin{proof}
For all $n \in \mathbb{N}$, let $(x_n)_{n\in \mathbb{N}},(q_n)_{n\in \mathbb{N}}$ be defined as in~\eqref{eq:def-12}. 
The sequence $(x_n)_{n\in \frac{1}{2}\mathbb{N}}$ is monotonically decreasing by construction, see~\eqref{eq:tau-optimal} and~\eqref{eq:inf-tau}. 
$x_n\in [0,\infty)$ for all $n\in \frac{1}{2}\mathbb{N}$ due to the non-negativity of the Petz divergence. 
Hence, $q_n=\exp(-(1-\alpha)x_n)\in (0,1]$ for all $n\in \frac{1}{2}\mathbb{N}$, and 
$(q_n)_{n\in \frac{1}{2}\mathbb{N}}$ is monotonically increasing. 

Let $\mathcal{N}_{A\rightarrow B}$, $\mathcal{N}_{B\rightarrow A}$ be defined as in~\eqref{eq:def-N}. 
For all $n\in \mathbb{N}$ 
\begin{align}
q_{n}\tau_{B}^{(n)\alpha}
=q_{n}\left(\mathcal{N}_{A\rightarrow B}(\sigma_{A}^{(n)})\right)^\alpha
&\stackrel{\eqref{eq:qn0}}{=} \tr_A[\rho_{AB}^\alpha \sigma_A^{(n)1-\alpha} ]
\label{eq:q-tau-alpha}\\
&\geq \lambda_{A,n}^{1-\alpha}\tr_A[\rho_{AB}^\alpha]
\geq \lambda_{A,n}^{1-\alpha}\lambda_B \rho_B^0
= \lambda_{A,n}^{1-\alpha}\lambda_B \tau_{B}^{(n)0},
\end{align}
which implies that 
\begin{equation}\label{eq:bn-an}
\lambda_{B,n}\geq \left(\frac{\lambda_B}{q_{n}}\right)^{\frac{1}{\alpha}} \lambda_{A,n}^{\frac{1-\alpha}{\alpha}}.
\end{equation}
For all $n\in \mathbb{N}$ 
\begin{align}
q_{n+1/2}\sigma_{A}^{(n+1)\alpha }
=q_{n+1/2}\left(\mathcal{N}_{B\rightarrow A}(\tau_{B}^{(n)})\right)^\alpha
&\stackrel{\eqref{eq:qn1}}{=} \tr_B[\rho_{AB}^\alpha \tau_B^{(n)1-\alpha} ]
\\
&\geq \lambda_{B,n}^{1-\alpha}\tr_B[\rho_{AB}^\alpha ]
\geq \lambda_{B,n}^{1-\alpha}\lambda_A \rho_A^0
=\lambda_{B,n}^{1-\alpha}\lambda_A \sigma_A^{(n+1)0}, 
\end{align}
which implies that 
\begin{equation}\label{eq:an-bn}
\lambda_{A,n+1}\geq \left(\frac{\lambda_A}{q_{n+1/2}}\right)^{\frac{1}{\alpha}} \lambda_{B,n}^{\frac{1-\alpha}{\alpha}}.
\end{equation}
We conclude that for all $n\in \mathbb{N}$
\begin{align}
\lambda_{A,n+1}&\stackrel{\substack{\eqref{eq:bn-an}\\ \eqref{eq:an-bn}}}{\geq} \left(\frac{\lambda_A}{q_{n+1/2}}\right)^{\frac{1}{\alpha}}
\left(\frac{\lambda_B}{q_{n}}\right)^{\frac{1-\alpha}{\alpha^2}} \lambda_{A,n}^{(\frac{1-\alpha}{\alpha})^2}
&\geq\left(\frac{\lambda_A^\alpha\lambda_B^{1-\alpha}}{q_{n+1/2}} \right)^{\frac{1}{\alpha^2}} \lambda_{A,n}^{(\frac{1-\alpha}{\alpha})^2}
\geq\left(\frac{\lambda_A^\alpha\lambda_B^{1-\alpha}}{q_{n+1}} \right)^{\frac{1}{\alpha^2}} \lambda_{A,n}^{(\frac{1-\alpha}{\alpha})^2}.
\label{eq:recursion-a}
\end{align}
The last two inequalities hold because $(q_n)_{n\in \frac{1}{2}\mathbb{N}}$ is monotonically increasing. 
\eqref{eq:recursion-a} is a recursion relation for $(\lambda_{A,n})_{n\in \mathbb{N}}$. 
Using again the monotonicity of $(q_n)_{n\in \frac{1}{2} \mathbb{N}}$, this recursion relation implies that for all $n\in \mathbb{N}$
\begin{align}\label{eq:lam-delta-a}
\lambda_{A,n}&\stackrel{\eqref{eq:recursion-a}}{\geq} \left(\frac{\lambda_A^\alpha \lambda_B^{1-\alpha}}{q_{n}} \right)^{\frac{1}{\alpha^2}\sum\limits_{k=0}^{n-1}(\frac{1-\alpha}{\alpha})^{2k}} 
\lambda_{A,0}^{(\frac{1-\alpha}{\alpha})^{2n}}
=\left(\frac{\lambda_A^\alpha \lambda_B^{1-\alpha}}{q_{n}} \right)^{\frac{1}{2\alpha-1}(1-(\frac{1-\alpha}{\alpha})^{2n})} 
\lambda_{A,0}^{(\frac{1-\alpha}{\alpha})^{2n}}
\eqqcolon c_{A,n}.
\end{align}
Combining~\eqref{eq:bn-an} with~\eqref{eq:lam-delta-a} implies that for all $n\in \mathbb{N}$
\begin{align}\label{eq:lam-delta-b}
\lambda_{B,n}
\geq \left(\frac{\lambda_B}{q_n}\right)^{\frac{1}{\alpha}} c_{A,n}^{\frac{1-\alpha}{\alpha}}
\eqqcolon c_{B,n}.
\end{align}

For all $n\in \mathbb{N}$
\begin{align}\label{eq:lam-an-a0}
\lambda_{A,n}
\stackrel{\eqref{eq:lam-delta-a}}{\geq} c_{A,n}
&\geq (\lambda_A^\alpha \lambda_B^{1-\alpha} )^{\frac{1}{2\alpha-1}(1-(\frac{1-\alpha}{\alpha})^{2n})} 
\lambda_{A,0}^{(\frac{1-\alpha}{\alpha})^{2n}}
\geq \min(1,(\lambda_A^\alpha \lambda_B^{1-\alpha} )^{\frac{1}{2\alpha-1}} ) 
\lambda_{A,0}
=c_{A}.
\end{align}
The second inequality holds because $q_n\in (0,1]$ for all $n\in \mathbb{N}$. 

For all $n\in \mathbb{N}$
\begin{align}
\lambda_{B,n}
\stackrel{\eqref{eq:lam-delta-b}}{\geq} c_{B,n}
&\geq \lambda_B^{\frac{1}{\alpha}}c_{A,n}^{\frac{1-\alpha}{\alpha}}
\stackrel{\eqref{eq:lam-an-a0}}{\geq} \lambda_B^{\frac{1}{\alpha}}c_{A}^{\frac{1-\alpha}{\alpha}}
=c_{B}.
\end{align}
The second inequality holds because $q_n\in (0,1]$ for all $n\in \mathbb{N}$. 
\end{proof}

The following lemma is a slightly modified variant of~\cite[Lemma~6.2]{beck2013convergence}. 

\begin{lem}[Sublinear convergence criterion]\label{lem:an}
Let $\gamma\in (0,\infty)$ and let $(A_n)_{n\in \mathbb{N}_{>0}}$ be a sequence of non-negative real numbers such that 
$\gamma A_{n+1}^2\leq A_n-A_{n+1}$ for all $n\in \mathbb{N}_{>0}$. 
Let $c\in [\frac{3}{2},\infty)$ be such that $A_n\leq \frac{c}{\gamma n}$ for all $n\in \{1,2\}$. 
Then, for all $n\in \mathbb{N}_{>0}$
\begin{align}\label{eq:an-induction}
A_{n}\leq \frac{c}{\gamma n}.
\end{align}
\end{lem}
\begin{proof}
We will prove the claim by induction over $n$. 

Base case: By assumption,~\eqref{eq:an-induction} is true for $n\in \{1,2\}$.

Induction step: Suppose~\eqref{eq:an-induction} holds for some $n\in \mathbb{N}_{\geq 2}$. Then,
\begin{align}
\gamma A_{n+1}^2+A_{n+1}
&\leq A_n\leq \frac{c}{\gamma n},
\label{eq:conv1}
\end{align}
where the last inequality follows from the induction hypothesis. 
Therefore, 
\begin{align}
A_{n+1}&\leq \frac{-1+\sqrt{1+4\frac{c}{n}}}{2\gamma}
\leq \frac{c}{\gamma (n+1)}.
\label{eq:conv2}
\end{align}
The first inequality in~\eqref{eq:conv2} follows from~\eqref{eq:conv1} and the fact that $c,\gamma\in (0,\infty)$.
The second inequality in~\eqref{eq:conv2} holds because $n\geq 2$ and $c\geq \frac{3}{2}$.
\end{proof}
\begin{rem}[Existence of constant $c$]\label{rem:def-c}
Lemma~\ref{lem:an} presupposes the existence of a suitable constant $c$. 
Such a constant generally exists because one may define 
$c\coloneqq \max(\frac{3}{2},\gamma A_1,2\gamma A_2)$.
\end{rem}

\subsection{Proof of Theorem~\ref{thm:sublinear}}\label{app:proof-thm}
\emph{Notation.} 
Before presenting the proof, we first clarify our notation for Fr\'echet derivatives. 
Consider $\mathcal{B}_A\coloneqq \{X_A\in \mathcal{L}(A):X_A^\dagger =X_A\}$ with the Schatten $\infty$-norm as a Banach space over $\mathbb{R}$. 
Similarly, consider $\mathcal{B}_B\coloneqq \{Y_B\in \mathcal{L}(B):Y_B^\dagger =Y_B\}$ with the Schatten $\infty$-norm as a Banach space over $\mathbb{R}$. 
If $U\subseteq \mathcal{B}_A$ is an open set and $f:U\rightarrow\mathcal{B}_B$ is a Fr\'echet differentiable function, then the Fr\'echet derivative of $f$ at $X_A\in U$ is denoted by 
$Df(X_A)\in \mathcal{L}(\mathcal{B}_A,\mathcal{B}_B)$.

\begin{proof}
Without loss of generality, suppose $\sigma_A^{(0)}\in \mathcal{S}_{\sim \rho_A}(A)$, see Remark~\ref{rem:restriction}.

In addition to the definitions in Theorem~\ref{thm:sublinear}, 
we define $\tau_B^{(n)},\sigma_A^{(n+1)}$ for all $n\in \mathbb{N}$ as in Definition~\ref{def:am-quantum}, 
and we define 
$(\lambda_{B,n})_{n\in \mathbb{N}}$ and $c_B$ as in Lemma~\ref{lem:spectrum-leq1}. 
Moreover, we define 
$(x_{n})_{n\in \frac{1}{2}\mathbb{N}}$ and $(q_{n})_{n\in \frac{1}{2}\mathbb{N}}$ as in~\eqref{eq:def-12}. 

We will first prove~\eqref{eq:thm-n1}. 
As explained in the proof of Lemma~\ref{lem:spectrum-leq1}, the sequence $(x_n)_{n\in \frac{1}{2}\mathbb{N}}$ is monotonically decreasing and $x_n\in [0,\infty)$ for all $n\in \frac{1}{2}\mathbb{N}$, and 
$(q_n)_{n\in \frac{1}{2}\mathbb{N}}$ is monotonically increasing and $q_n\in (0,1]$ for all $n\in \frac{1}{2}\mathbb{N}$. 

For all $n\in \mathbb{N}$
\begin{align}
x_{n+1/2}-x_{n+1}
&=D_\alpha (\rho_{AB}\| \sigma_A^{(n+1)}\otimes \tau_B^{(n)})
-D_\alpha (\rho_{AB}\| \sigma_A^{(n+1)}\otimes \tau_B^{(n+1)})
=D_\alpha (\tau_B^{(n+1)}\| \tau_B^{(n)}).
\label{eq:xn-dn}
\end{align}
The last equality follows from the quantum Sibson identity~\eqref{eq:sibson}. 
Since $q_n=\exp((\alpha-1)x_n)$ for all $n\in \frac{1}{2}\mathbb{N}$, it follows from~\eqref{eq:xn-dn} that for all $n\in \mathbb{N}$
\begin{align}
\frac{q_{n+1/2}}{q_{n+1}}
&=Q_\alpha(\tau_B^{(n+1)}\| \tau_B^{(n)}).
\label{eq:q-frac-tau}
\end{align}

Without loss of generality, suppose $\rho_A>0$ and $\rho_B>0$. 
(If this does not hold, the same proof works if one restricts the Hilbert spaces $A$ and $B$ to the support of $\rho_A$ and the support of $\rho_B$, respectively.) 
Let us define the function
\begin{align}\label{eq:def-f}
f:\mathcal{S}_{>0}(A)\times \mathcal{S}_{>0}(B)\rightarrow\mathbb{R}, 
(\sigma_A,\tau_B)\mapsto -Q_\alpha(\rho_{AB}\| \sigma_A\otimes \tau_B)
=-\tr[\rho_{AB}^\alpha (\sigma_A\otimes \tau_B)^{1-\alpha}].
\end{align}
We will now examine the Fr\'echet derivative of $f$. 
Let $g:\mathcal{S}_{>0}(B)\rightarrow\mathcal{L}(B),\tau_B\mapsto \tau_B^{1-\alpha}$. 
The Fr\'echet derivative of $g$ at $\tau_B\in \mathcal{S}_{>0}(B)$ is 
$Dg(\tau_B)(Y_B)=g^{[1]}(\tau_B)\odot Y_B$ for all self-adjoint $Y_B\in \mathcal{L}(B)$, 
where $\odot$ denotes the Hadamard product taken in an eigenbasis of $\tau_B$, 
and $g^{[1]}$ denotes the first-order divided difference of $g$~\cite{bhatia2007positive}. 
Therefore, the partial Fr\'echet derivative with respect to the second argument of $f$ at 
$(\sigma_A,\tau_B)\in \mathcal{S}_{>0}(A)\times \mathcal{S}_{>0}(B)$ is
\begin{align}
D_2f (\sigma_A,\tau_B)(Y_B)
&=-\tr[ \tr_A[\rho_{AB}^\alpha \sigma_A^{1-\alpha}]  Dg(\tau_B)(Y_B) ]
\label{eq:nabla1f}
\end{align}
for all self-adjoint $Y_B\in \mathcal{L}(B)$. 
Similarly, let $\tilde{g}:\mathcal{S}_{>0}(A)\rightarrow\mathcal{L}(A),\sigma_A\mapsto \sigma_A^{1-\alpha}$. 
The Fr\'echet derivative of $\tilde{g}$ at $\sigma_A\in \mathcal{S}_{>0}(A)$ is 
$D\tilde{g}(\sigma_A)(X_A)=\tilde{g}^{[1]}(\sigma_A)\tilde{\odot} X_A$ for all self-adjoint $X_A\in \mathcal{L}(A)$, 
where $\tilde{\odot}$ denotes the Hadamard product taken in an eigenbasis of $\sigma_A$, 
and $\tilde{g}^{[1]}$ denotes the first-order divided difference of $\tilde{g}$~\cite{bhatia2007positive}. 
Therefore, the partial Fr\'echet derivative with respect to the first argument of $f$ at 
$(\sigma_A,\tau_B)\in\mathcal{S}_{>0}(A)\times \mathcal{S}_{>0}(B)$ is
\begin{align}\label{eq:D1f}
D_1f (\sigma_A,\tau_B)(X_A)
&=-\tr[ \tr_B[\rho_{AB}^\alpha \tau_B^{1-\alpha}] D\tilde{g}(\sigma_A)(X_A) ]
\end{align}
for all self-adjoint $X_A\in \mathcal{L}(A)$.
The partial Fr\'echet derivatives are continuous, i.e., 
$(\sigma_A,\tau_B)\mapsto D_if(\sigma_A,\tau_B)$ is continuous on $\mathcal{S}_{>0}(A)\times \mathcal{S}_{>0}(B)$ for any $i\in \{1,2\}$. 
Therefore, $f$ is Fr\'echet differentiable and the Fr\'echet derivative of $f$ is determined by its partial derivatives as
\begin{equation}\label{eq:frechet-partial}
Df(\sigma_A,\tau_B)(X_A,Y_B)
=D_1f(\sigma_A,\tau_B)(X_A)+D_2f(\sigma_A,\tau_B)(Y_B)
\end{equation}
for all self-adjoint $X_A\in \mathcal{L}(A),Y_B\in \mathcal{L}(B)$.

We have
\begin{align}
\argmin_{(\sigma_A,\tau_B)\in \mathcal{S}_{>0}(A)\times \mathcal{S}_{>0}(B)} 
f(\sigma_A,\tau_B)
&=\argmax_{(\sigma_A,\tau_B)\in \mathcal{S}_{>0}(A)\times \mathcal{S}_{>0}(B)} 
Q_\alpha(\rho_{AB}\| \sigma_A\otimes \tau_B)
\\
&=\argmin_{(\sigma_A,\tau_B)\in \mathcal{S}(A)\times \mathcal{S}(B)} 
D_\alpha(\rho_{AB}\| \sigma_A\otimes \tau_B)
\label{eq:argmin-f}
\end{align}
where~\eqref{eq:argmin-f} holds due to~\eqref{eq:unique}. 
Let $(\hat{\sigma}_A,\hat{\tau}_B)$ be the unique element in this set, see~\eqref{eq:unique}. 
Let 
\begin{align}
\hat{q}&\coloneqq Q_\alpha(\rho_{AB}\| \hat{\sigma}_A\otimes\hat{\tau}_B)
= -f(\hat{\sigma}_A,\hat{\tau}_B)
=\exp((\alpha-1)I_\alpha^{\downarrow\downarrow}(A:B)_\rho).
\end{align}

For all $n\in \mathbb{N}$
\begin{align}
0&\leq \hat{q}-q_{n+1}
\leq \hat{q}-q_{n+1/2}
\label{eq:am0}\\
&=f(\sigma_A^{(n+1)},\tau_B^{(n)})-f(\hat{\sigma}_A,\hat{\tau}_B)
\label{eq:am1}\\
&\leq Df(\sigma_A^{(n+1)},\tau_B^{(n)}) (\sigma_A^{(n+1)}-\hat{\sigma}_A,\tau_B^{(n)}-\hat{\tau}_B)
\label{eq:am12}\\
&= D_1 f(\sigma_A^{(n+1)},\tau_B^{(n)}) (\sigma_A^{(n+1)}-\hat{\sigma}_A) 
+ D_2 f(\sigma_A^{(n+1)},\tau_B^{(n)}) (\tau_B^{(n)}-\hat{\tau}_B)
\label{eq:am2}\\
&=D_2 f(\sigma_A^{(n+1)},\tau_B^{(n)}) (\tau_B^{(n)}-\hat{\tau}_B)
\label{eq:am3}\\
&= D_2 f(\sigma_A^{(n+1)},\tau_B^{(n)})(\tau_B^{(n)}-\hat{\tau}_B)
- D_2 f (\sigma_A^{(n)},\tau_B^{(n)})(\tau_B^{(n)}-\hat{\tau}_B)
\label{eq:am5}
\\
&=\tr[(\tr_A[\rho_{AB}^\alpha \sigma_A^{(n)1-\alpha}]
-\tr_A[\rho_{AB}^\alpha \sigma_A^{(n+1)1-\alpha}])
Dg(\tau_B^{(n)})(\tau_B^{(n)}-\hat{\tau}_B)]
\label{eq:am6}
\\
&=\tr[(q_{n} \tau_B^{(n)\alpha}- q_{n+1} \tau_B^{(n+1)\alpha} ) 
Dg(\tau_B^{(n)}) (\tau_B^{(n)}-\hat{\tau}_B) ]
\label{eq:am7}\\
&\leq \lVert q_{n} \tau_B^{(n)\alpha}- q_{n+1} \tau_B^{(n+1)\alpha} \rVert_2 \cdot 
\lVert Dg(\tau_B^{(n)}) (\tau_B^{(n)}-\hat{\tau}_B)
\rVert_2.
\label{eq:am8}
\end{align}
\eqref{eq:am0} holds because $(q_n)_{n\in \frac{1}{2}\mathbb{N}}$ is monotonically increasing. 
\eqref{eq:am12} follows from the joint convexity of $f$~\cite{burri2024properties}. 
\eqref{eq:am2} follows from~\eqref{eq:frechet-partial}.
\eqref{eq:am3} holds because according to~\eqref{eq:tau-optimal}, 
$\sigma_A^{(n+1)}$ is the unique global minimizer of $f(\cdot,\tau_B^{(n)})$.
\eqref{eq:am5} holds because according to~\eqref{eq:tau-optimal}, 
$\tau_B^{(n)}$ is the unique global minimizer of $f(\sigma_A^{(n)},\cdot )$. 
\eqref{eq:am6} follows from~\eqref{eq:nabla1f}. 
\eqref{eq:am7} follows from~\eqref{eq:q-tau-alpha}. 
\eqref{eq:am8} follows from the Cauchy-Schwarz inequality for the Hilbert-Schmidt inner product.

The first factor in~\eqref{eq:am8} can be bounded from above as follows.
\begin{align}
\big\lVert q_{n} \tau_B^{(n)\alpha}- q_{n+1} \tau_B^{(n+1)\alpha} \big\rVert_2
&=\big\lVert q_{n}(\tau_B^{(n)\alpha}-\tau_B^{(n+1)\alpha})-(q_{n+1}-q_{n})\tau_B^{(n+1)\alpha} \big\rVert_2
\\
&\leq q_{n}\big\lVert \tau_B^{(n)\alpha}- \tau_B^{(n+1)\alpha} \big\rVert_2
+(q_{n+1}-q_{n})\lVert \tau_B^{(n+1)\alpha}\rVert_2
\label{eq:first-1}\\
&\leq q_{n}\big\lVert \tau_B^{(n)\alpha}- \tau_B^{(n+1)\alpha} \big\rVert_{1/\alpha}
+(q_{n+1}-q_{n})\lVert \tau_B^{(n+1)\alpha}\rVert_{1/\alpha}
\label{eq:first-2}\\
&\leq q_{n}\sqrt{\frac{4\alpha}{1-\alpha}\left(1-Q_\alpha (\tau_B^{(n+1)}\| \tau_B^{(n)}) \right) }
+q_{n+1}-q_{n}
\label{eq:first-3}\\
&=q_{n}\sqrt{\frac{4\alpha}{1-\alpha}\left(1-\frac{q_{n+1/2}}{q_{n+1}} \right) }
+q_{n+1}-q_{n}
\label{eq:first-4}\\
&\leq q_{n+1}\left(\sqrt{\frac{4\alpha}{1-\alpha}\left(1-\frac{q_n}{q_{n+1}}\right)} +\left(1-\frac{q_n}{q_{n+1}}\right) \right)
\label{eq:first-5}\\
&\leq q_{n+1}\left(\sqrt{\frac{4\alpha}{1-\alpha}}+1\right) \sqrt{1-\frac{q_n}{q_{n+1}}} 
\label{eq:first-6}\\
&= \left(\sqrt{\frac{4\alpha}{1-\alpha}}+1\right) \sqrt{q_{n+1}(q_{n+1}-q_n)} 
\label{eq:first-7}
\end{align}
\eqref{eq:first-1} follows from the subadditivity of norms. 
\eqref{eq:first-2} holds due to the monotonicity of the Schatten norms. 
\eqref{eq:first-3} follows from the R\'enyi Pinsker inequality~\eqref{eq:pinsker}.
\eqref{eq:first-4} follows from~\eqref{eq:q-frac-tau}.
\eqref{eq:first-5} holds because $(q_n)_{n\in \frac{1}{2}\mathbb{N}}$ is monotonically increasing. 
\eqref{eq:first-6} holds because $t\leq \sqrt{t}$ for all $t\in [0,1]$.

The second factor in~\eqref{eq:am8} can be bounded from above as follows.
\begin{align}
\lVert Dg(\tau_B^{(n)}) (\tau_B^{(n)}-\hat{\tau}_B) \rVert_2
&\leq \lVert Dg(\tau_B^{(n)}) (\tau_B^{(n)}) \rVert_{2}+
\lVert Dg(\tau_B^{(n)}) (\hat{\tau}_B) \rVert_{2}
\label{eq:dtau0}\\
&\leq 2\sup_{Y_B\in \mathcal{S}(B)} 
\lVert Dg(\tau_B^{(n)}) (Y_B) \rVert_{2}
=2\sup_{\substack{Y_B\in \mathcal{L}(B): \\ Y_B\geq 0, \lVert Y_B\rVert_1 =1}} 
\lVert Dg(\tau_B^{(n)}) (Y_B) \rVert_{2}
\label{eq:dtau11}\\
&\leq 2\sup_{\substack{Y_B\in \mathcal{L}(B):\\ Y_B\geq 0, \lVert Y_B\rVert_2\leq 1 }}
\lVert Dg(\tau_B^{(n)}) (Y_B) \rVert_{2}
\label{eq:dtau2}\\
&\leq 2\sup_{\substack{Y_B\in \mathcal{L}(B):\\ Y_B^\dagger =Y_B, \lVert Y_B\rVert_2\leq 1 }}
\lVert Dg(\tau_B^{(n)}) (Y_B) \rVert_{2}
\label{eq:dtau21}\\
&\leq 2 \lVert (1-\alpha)\tau_B^{(n)-\alpha} \rVert_{\infty}
\label{eq:dtau4}\\
&=\frac{2}{\lambda_{B,n}^{\alpha}} (1-\alpha)
\label{eq:dtau5}
\end{align}
\eqref{eq:dtau0} follows from the subadditivity of norms. 
\eqref{eq:dtau2} holds due to the monotonicity of the Schatten norms. 
\eqref{eq:dtau4} holds because the Schatten 2-norm is unitarily invariant and $g(\tau_B)=\tau_B^{1-\alpha}$ is operator monotone \cite{bhatia1996matrix,bhatia2007positive}.
\eqref{eq:dtau5} is trivially true by the definition of $\lambda_{B,n}$ in~\eqref{eq:lambda-bn}.

Substituting these bounds for the first and the second factor into~\eqref{eq:am8} implies that for all $n\in \mathbb{N}$
\begin{align}
0\leq \hat{q}-q_{n+1}
&\leq \frac{2}{\lambda_{B,n}^{\alpha}} \sqrt{(1-\alpha)q_{n+1}\left(q_{n+1}-q_n\right)} 
\left(\sqrt{4\alpha}+\sqrt{1-\alpha}\right)
\\
&\leq \frac{2\sqrt{5}}{\lambda_{B,n}^{\alpha}} \sqrt{(1-\alpha)q_{n+1}\left(q_{n+1}-q_n\right)}.
\label{eq:alpha-term}
\end{align}
\eqref{eq:alpha-term} holds because 
$\sqrt{4t}+\sqrt{1-t}\leq \sqrt{5}$ for all $t\in [\frac{1}{2},1]$. 
The division of~\eqref{eq:alpha-term} by $q_{n+1}$ implies that for all $n\in \mathbb{N}$
\begin{align}\label{eq:q-qn1}
0\leq \frac{\hat{q}}{q_{n+1}}-1
\leq \frac{2\sqrt{5}}{\lambda_{B,n}^{\alpha}} \sqrt{(1-\alpha)\left(1-\frac{q_n}{q_{n+1}}\right)}.
\end{align} 

For all $n\in \mathbb{N}$
\begin{align}\label{eq:xn-xn1}
x_n-x_{n+1}
=-\frac{1}{1-\alpha}\log \left(1+\frac{q_n}{q_{n+1}}-1 \right)
\geq -\frac{1}{1-\alpha}\left(\frac{q_n}{q_{n+1}}-1 \right)
=\frac{1}{1-\alpha}\left(1-\frac{q_n}{q_{n+1}} \right).
\end{align}
For the inequality in~\eqref{eq:xn-xn1}, we have used that $\log(1+t)\leq t$ for all $t\in (-1,0]$ and that $\frac{q_n}{q_{n+1}}\in (0,1]$ since $(q_n)_{n\in \mathbb{N}}$ is monotonically increasing. 
For all $n\in \mathbb{N}$
\begin{align}
0\leq x_{n+1}-I_\alpha^{\downarrow\downarrow}(A:B)_\rho
&=\frac{1}{1-\alpha}\log\left(1+\frac{\hat{q}}{q_{n+1}}-1\right)
\\
&\leq \frac{1}{1-\alpha}\left(\frac{\hat{q}}{q_{n+1}}-1\right)
\label{eq:xn-x-1}\\
&\leq \frac{2\sqrt{5}}{\lambda_{B,n}^{\alpha}} \sqrt{\frac{1}{1-\alpha}\left(1-\frac{q_n}{q_{n+1}}\right)}
\label{eq:xn-x-12}
\\
&\leq \frac{2\sqrt{5}}{\lambda_{B,n}^{\alpha}} \sqrt{x_n-x_{n+1}}
\label{eq:xn-x-2}
\\
&\leq \frac{2\sqrt{5}}{c_{B}^{\alpha}} \sqrt{x_n-x_{n+1}} 
=c_0\sqrt{x_n-x_{n+1}}.
\label{eq:xn-x-3}
\end{align}
\eqref{eq:xn-x-1} holds because $\log(1+t)\leq t$ for all $t\in [0,\infty)$. 
\eqref{eq:xn-x-12} follows from~\eqref{eq:q-qn1}.
\eqref{eq:xn-x-2} follows from~\eqref{eq:xn-xn1}. 
The inequality in~\eqref{eq:xn-x-3} follows from Lemma~\ref{lem:spectrum-leq1}. 
The equality in~\eqref{eq:xn-x-3} follows from~\eqref{eq:def-sublinear-c0} and~\eqref{eq:def-cb}.
This proves~\eqref{eq:thm-n1}.

We will now prove~\eqref{eq:thm-n2}. 
Let $A_n\coloneqq x_n-I_\alpha^{\downarrow\downarrow}(A:B)_\rho$ for all $n\in \mathbb{N}$. 
Let $\gamma\coloneqq c_0^{-2}$. 
By~\eqref{eq:thm-n1}, 
$\gamma A_{n+1}^2\leq A_n-A_{n+1}$ for all $n\in \mathbb{N}$. 
Let $c\coloneqq \max(\frac{3}{2},2\gamma A_0)$. 
Since the sequence $(A_n)_{n\in \mathbb{N}}$ is monotonically decreasing and non-negative, 
\begin{align}\label{eq:proof-c}
c=\max\left(\frac{3}{2},2\gamma A_0,2\gamma A_1,2\gamma A_2\right)
\geq \max\left(\frac{3}{2},\gamma A_1,2\gamma A_2\right).
\end{align} 

We conclude that for all $n\in \mathbb{N}_{>0}$
\begin{align}
\lvert x_{n}-I_\alpha^{\downarrow\downarrow}(A:B)_\rho \rvert
=A_n
\leq \frac{c}{\gamma n}
&=\max\left(\frac{3}{2} c_0^{2},2(x_0-I_\alpha^{\downarrow\downarrow}(A:B)_\rho)\right) \frac{1}{n}
\label{eq:proof-xn-an1}
\\
&\leq \max\left(\frac{3}{2} c_0^{2},2x_0\right) \frac{1}{n}.
\label{eq:proof-xn-an2}
\end{align}
The inequality in~\eqref{eq:proof-xn-an1} follows from Lemma~\ref{lem:an} and~\eqref{eq:proof-c}. 
\eqref{eq:proof-xn-an2} follows from the non-negativity of the doubly minimized PRMI~\cite{burri2024properties}.
\end{proof}

\bibliographystyle{arxiv_fullname}
\bibliography{bibfile}

\begin{thebibliography}{10}

\bibitem{mckinlay2020decomposition}
Alexander McKinlay and Marco Tomamichel.
\newblock {Decomposition rules for quantum R\'enyi mutual information with an
  application to information exclusion relations}.
\newblock {\em Journal of Mathematical Physics}, 61(7), 2020.
\newblock
  \texttt{\href{http://dx.doi.org/10.1063/1.5143862}{DOI:\,10.1063/1.5143862}}.

\bibitem{beigi2013sandwiched}
Salman Beigi.
\newblock {Sandwiched R\'enyi divergence satisfies data processing inequality}.
\newblock {\em Journal of Mathematical Physics}, 54(12), 2013.
\newblock
  \texttt{\href{http://dx.doi.org/10.1063/1.4838855}{DOI:\,10.1063/1.4838855}}.

\bibitem{leditzky2016strong}
Felix Leditzky, Mark~M. Wilde, and Nilanjana Datta.
\newblock {Strong converse theorems using R\'enyi entropies}.
\newblock {\em Journal of Mathematical Physics}, 57(8), 2016.
\newblock
  \texttt{\href{http://dx.doi.org/10.1063/1.4960099}{DOI:\,10.1063/1.4960099}}.

\bibitem{cheng2023tight}
Hao-Chung Cheng and Li~Gao.
\newblock {Tight One-Shot Analysis for Convex Splitting with Applications in
  Quantum Information Theory}, 2023.
\newblock
  \texttt{\href{http://dx.doi.org/10.48550/arXiv.2304.12055}{DOI:\,10.48550/arXiv.2304.12055}}.

\bibitem{wilde2014strong}
Mark~M. Wilde, Andreas Winter, and Dong Yang.
\newblock {Strong Converse for the Classical Capacity of Entanglement-Breaking
  and Hadamard Channels via a Sandwiched R\'enyi Relative Entropy}.
\newblock {\em Communications in Mathematical Physics}, 331(2):593--622, 2014.
\newblock
  \texttt{\href{http://dx.doi.org/10.1007/s00220-014-2122-x}{DOI:\,10.1007/s00220-014-2122-x}}.

\bibitem{li2022reliability}
Ke~Li and Yongsheng Yao.
\newblock {Reliability Function of Quantum Information Decoupling via the
  Sandwiched R\'enyi Divergence}.
\newblock {\em Communications in Mathematical Physics}, 405(7), 2024.
\newblock
  \texttt{\href{http://dx.doi.org/10.1007/s00220-024-05029-z}{DOI:\,10.1007/s00220-024-05029-z}}.

\bibitem{li2024operational}
Ke~Li and Yongsheng Yao.
\newblock {Operational Interpretation of the Sandwiched R\'enyi Divergence of
  Order 1/2 to 1 as Strong Converse Exponents}.
\newblock {\em Communications in Mathematical Physics}, 405(22), 2024.
\newblock
  \texttt{\href{http://dx.doi.org/10.1007/s00220-023-04890-8}{DOI:\,10.1007/s00220-023-04890-8}}.

\bibitem{burri2024properties2}
Laura Burri.
\newblock {Doubly minimized sandwiched R\'enyi mutual information: Properties
  and operational interpretation from strong converse exponent}, 2024.
\newblock
  \texttt{\href{http://dx.doi.org/10.48550/arXiv.2406.03213}{DOI:\,10.48550/arXiv.2406.03213}}.

\bibitem{gupta2014multiplicativity}
Manish~K. Gupta and Mark~M. Wilde.
\newblock {Multiplicativity of Completely Bounded $p$-Norms Implies a Strong
  Converse for Entanglement-Assisted Capacity}.
\newblock {\em Communications in Mathematical Physics}, 334(2):867--887, 2014.
\newblock
  \texttt{\href{http://dx.doi.org/10.1007/s00220-014-2212-9}{DOI:\,10.1007/s00220-014-2212-9}}.

\bibitem{berta2015renyi}
Mario Berta, Kaushik~P. Seshadreesan, and Mark~M. Wilde.
\newblock R\'enyi generalizations of the conditional quantum mutual
  information.
\newblock {\em Journal of Mathematical Physics}, 56(2), 2015.
\newblock
  \texttt{\href{http://dx.doi.org/10.1063/1.4908102}{DOI:\,10.1063/1.4908102}}.

\bibitem{mosonyi2015coding}
Mil\'an Mosonyi.
\newblock {Coding Theorems for Compound Problems via Quantum R\'enyi
  Divergences}.
\newblock {\em IEEE Transactions on Information Theory}, 61(6):2997--3012,
  2015.
\newblock
  \texttt{\href{http://dx.doi.org/10.1109/TIT.2015.2417877}{DOI:\,10.1109/TIT.2015.2417877}}.

\bibitem{mosonyi2017strong}
Mil\'an Mosonyi and Tomohiro Ogawa.
\newblock {Strong Converse Exponent for Classical-Quantum Channel Coding}.
\newblock {\em Communications in Mathematical Physics}, 355(1):373--426, 2017.
\newblock
  \texttt{\href{http://dx.doi.org/10.1007/s00220-017-2928-4}{DOI:\,10.1007/s00220-017-2928-4}}.

\bibitem{hayashi2016correlation}
Masahito Hayashi and Marco Tomamichel.
\newblock {Correlation detection and an operational interpretation of the
  R\'enyi mutual information}.
\newblock {\em Journal of Mathematical Physics}, 57(102201), 2016.
\newblock
  \texttt{\href{http://dx.doi.org/10.1063/1.4964755}{DOI:\,10.1063/1.4964755}}.

\bibitem{kudlerflam2023renyi}
Jonah Kudler-Flam, Laimei Nie, and Akash Vijay.
\newblock R\'enyi mutual information in quantum field theory, tensor networks,
  and gravity.
\newblock {\em Journal of High Energy Physics}, 2024(6), 2024.
\newblock
  \texttt{\href{http://dx.doi.org/10.1007/JHEP06(2024)195}{DOI:\,10.1007/JHEP06(2024)195}}.

\bibitem{kudlerflam2023renyi1}
Jonah Kudler-Flam.
\newblock {R\'enyi Mutual Information in Quantum Field Theory}.
\newblock {\em Physical Review Letters}, 130(021603), 2023.
\newblock
  \texttt{\href{http://dx.doi.org/10.1103/PhysRevLett.130.021603}{DOI:\,10.1103/PhysRevLett.130.021603}}.

\bibitem{berta2021composite}
Mario Berta, Fernando G. S.~L. Brand\~{a}o, and Christoph Hirche.
\newblock {On Composite Quantum Hypothesis Testing}.
\newblock {\em Communications in Mathematical Physics}, 385(1):55--77, 2021.
\newblock
  \texttt{\href{http://dx.doi.org/10.1007/s00220-021-04133-8}{DOI:\,10.1007/s00220-021-04133-8}}.

\bibitem{burri2024properties}
Laura Burri.
\newblock {Doubly minimized Petz R\'enyi mutual information: Properties and
  operational interpretation from direct exponent}, 2024.
\newblock
  \texttt{\href{http://dx.doi.org/10.48550/arXiv.2406.01699}{DOI:\,10.48550/arXiv.2406.01699}}.

\bibitem{tomamichel2018operational}
Marco Tomamichel and Masahito Hayashi.
\newblock {Operational Interpretation of R\'enyi Information Measures via
  Composite Hypothesis Testing Against Product and Markov Distributions}.
\newblock {\em IEEE Transactions on Information Theory}, 64(2):1064--1082,
  2018.
\newblock
  \texttt{\href{http://dx.doi.org/10.1109/TIT.2017.2776900}{DOI:\,10.1109/TIT.2017.2776900}}.

\bibitem{lapidoth2019two}
Amos Lapidoth and Christoph Pfister.
\newblock {Two Measures of Dependence}.
\newblock {\em Entropy}, 21(778), 2019.
\newblock
  \texttt{\href{http://dx.doi.org/10.3390/e21080778}{DOI:\,10.3390/e21080778}}.

\bibitem{burri2025minreflectedentropydoubly}
Laura Burri.
\newblock {Min-reflected entropy = doubly minimized Petz R\'enyi mutual
  information of order 1/2}, 2025.
\newblock
  \texttt{\href{http://dx.doi.org/10.48550/arXiv.2502.18433}{DOI:\,10.48550/arXiv.2502.18433}}.

\bibitem{kamatsuka2024algorithms}
Akira Kamatsuka, Koki Kazama, and Takahiro Yoshida.
\newblock {Alternating Optimization Approach for Computing $\alpha$-Mutual
  Information and $\alpha$-Capacity}, 2024.
\newblock
  \texttt{\href{http://dx.doi.org/10.48550/arXiv.2404.10950}{DOI:\,10.48550/arXiv.2404.10950}}.

\bibitem{tsai2024linearconvergencehilbertsprojective}
Chung-En Tsai, Guan-Ren Wang, Hao-Chung Cheng, and Yen-Huan Li.
\newblock {Linear Convergence in Hilbert's Projective Metric for Computing
  Augustin Information and a R\'enyi Information Measure}, 2024.
\newblock
  \texttt{\href{http://dx.doi.org/10.48550/arXiv.2409.02640}{DOI:\,10.48550/arXiv.2409.02640}}.

\bibitem{both2021rate}
Jakub~Wiktor Both.
\newblock {On the rate of convergence of alternating minimization for
  non-smooth non-strongly convex optimization in Banach spaces}.
\newblock {\em Optimization Letters}, 16(2):729--743, 2021.
\newblock
  \texttt{\href{http://dx.doi.org/10.1007/s11590-021-01753-w}{DOI:\,10.1007/s11590-021-01753-w}}.

\bibitem{petz1986quasi}
D\'enes Petz.
\newblock Quasi-entropies for finite quantum systems.
\newblock {\em Reports on Mathematical Physics}, 23(1):57--65, 1986.
\newblock
  \texttt{\href{http://dx.doi.org/10.1016/0034-4877(86)90067-4}{DOI:\,10.1016/0034-4877(86)90067-4}}.

\bibitem{carlen2017remainder}
Eric~A. Carlen.
\newblock {A remainder term for H\"older's inequality for matrices and quantum
  entropy inequalities}.
\newblock {\em Archiv der Mathematik}, 109:365–371, 2017.
\newblock
  \texttt{\href{http://dx.doi.org/10.1007/s00013-017-1066-8}{DOI:\,10.1007/s00013-017-1066-8}}.

\bibitem{lemmens2012nonlinear}
Bas Lemmens and Roger Nussbaum.
\newblock {\em {Nonlinear Perron--Frobenius Theory}}.
\newblock Cambridge Tracts in Mathematics. Cambridge University Press, 2012.
\newblock
  \texttt{\href{http://dx.doi.org/10.1017/CBO9781139026079}{DOI:\,10.1017/CBO9781139026079}}.

\bibitem{lemmens2013birkhoffs}
Bas Lemmens and Roger Nussbaum.
\newblock {Birkhoff's version of Hilbert's metric and its applications in
  analysis}, 2013.
\newblock
  \texttt{\href{http://dx.doi.org/10.48550/arXiv.1304.7921}{DOI:\,10.48550/arXiv.1304.7921}}.

\bibitem{reeb2011hilbert}
David Reeb, Michael~J. Kastoryano, and Michael~M. Wolf.
\newblock Hilbert’s projective metric in quantum information theory.
\newblock {\em Journal of Mathematical Physics}, 52(8), 2011.
\newblock
  \texttt{\href{http://dx.doi.org/10.1063/1.3615729}{DOI:\,10.1063/1.3615729}}.

\bibitem{bauschke2017convex}
Heinz~H. Bauschke and Patrick~L. Combettes.
\newblock {\em {Convex Analysis and Monotone Operator Theory in Hilbert
  Spaces}}.
\newblock CMS Books in Mathematics. Springer, 2020.
\newblock
  \texttt{\href{http://dx.doi.org/10.1007/978-3-319-48311-5}{DOI:\,10.1007/978-3-319-48311-5}}.

\bibitem{zaslavski2020projected}
Alexander~J. Zaslavski.
\newblock {\em {The Projected Subgradient Algorithm in Convex Optimization}}.
\newblock SpringerBriefs in Optimization. Springer, 2020.
\newblock
  \texttt{\href{http://dx.doi.org/10.1007/978-3-030-60300-7}{DOI:\,10.1007/978-3-030-60300-7}}.

\bibitem{zaslavski2022optimization}
Alexander~J. Zaslavski.
\newblock {\em {Optimization in Banach Spaces}}.
\newblock SpringerBriefs in Optimization. Springer, 2022.
\newblock
  \texttt{\href{http://dx.doi.org/10.1007/978-3-031-12644-4}{DOI:\,10.1007/978-3-031-12644-4}}.

\bibitem{tomamichel2016quantum}
Marco Tomamichel.
\newblock {\em {Quantum Information Processing with Finite Resources}}.
\newblock Springer, 2016.
\newblock
  \texttt{\href{http://dx.doi.org/10.1007/978-3-319-21891-5}{DOI:\,10.1007/978-3-319-21891-5}}.

\bibitem{birkhoff1957extensions}
Garrett Birkhoff.
\newblock {Extensions of Jentzsch's Theorem}.
\newblock {\em Transactions of the American Mathematical Society},
  85(1):219--227, 1957.
\newblock
  \texttt{\href{http://dx.doi.org/10.2307/1992971}{DOI:\,10.2307/1992971}}.

\bibitem{hopf1963inequality}
Eberhard Hopf.
\newblock {An Inequality for Positive Linear Integral Operators}.
\newblock {\em Journal of Mathematics and Mechanics}, 12(5):683--692, 1963.
\newblock Available online: \url{https://www.jstor.org/stable/24900876}.

\bibitem{beck2013convergence}
Amir Beck and Luba Tetruashvili.
\newblock {On the Convergence of Block Coordinate Descent Type Methods}.
\newblock {\em SIAM Journal on Optimization}, 23:2037--2060, 2013.
\newblock
  \texttt{\href{http://dx.doi.org/10.1137/120887679}{DOI:\,10.1137/120887679}}.

\bibitem{bhatia2007positive}
Rajendra Bhatia.
\newblock {\em {Positive Definite Matrices}}.
\newblock Princeton Series in Applied Mathematics. Princeton University Press,
  2007.
\newblock
  \texttt{\href{http://dx.doi.org/10.1515/9781400827787}{DOI:\,10.1515/9781400827787}}.

\bibitem{bhatia1996matrix}
Rajendra Bhatia.
\newblock {\em {Matrix Analysis}}.
\newblock Graduate Texts in Mathematics. Springer, 1997.
\newblock
  \texttt{\href{http://dx.doi.org/10.1007/978-1-4612-0653-8}{DOI:\,10.1007/978-1-4612-0653-8}}.

\end{thebibliography}

\end{document}